\documentclass[review]{elsarticle}

\usepackage{lineno,hyperref}

\usepackage{amsmath, amsfonts, amssymb, float, enumerate, bm, graphicx, mathtools, mathrsfs, appendix}

\usepackage{kbordermatrix}
\renewcommand{\kbldelim}{(}
\renewcommand{\kbrdelim}{)}

\modulolinenumbers[5]

\journal{Journal of Games and Economic Behaviour}

\DeclareMathOperator*{\argmin}{arg\,min}
\DeclareMathOperator*{\argmax}{arg\,max}

 \newcounter{eg}[section]
\renewcommand{\theeg}{\arabic{section}.\arabic{eg}}
\newenvironment{examp}[1][]{\refstepcounter{eg}
\par\medskip \noindent
\textit{Example~\theeg. #1} \rmfamily}{\hfill $\square$   \hspace{-4.5pt} \vspace{6pt}}

\usepackage{amsmath, amssymb, bbm, xspace}
\usepackage{epsfig}
\usepackage{longtable}
\usepackage{color}
\usepackage{mathrsfs}
\usepackage{comment}
\usepackage{ifthen}
\newboolean{showcomments}
\setboolean{showcomments}{true}
\usepackage{courier}
\usepackage{xcolor}
\usepackage{todonotes}
\usepackage[framemethod=tikz]{mdframed}
\usepackage{lineno}
\newmdenv[leftmargin=\dimexpr-0.4em, innerleftmargin=0.5em,
rightmargin=\dimexpr-0.4em, innerrightmargin=0.5em,
linewidth=2pt,linecolor=red, topline=false, bottomline=false,
innertopmargin=0pt,innerbottommargin=0pt,skipbelow=0pt,skipabove=0pt,%
]{notex}

\newenvironment{note}%
{\vskip\dimexpr\dp\strutbox-\prevdepth\relax\notex\strut\ignorespaces}%
{\xdef\notetpd{\the\prevdepth}\endnotex\vskip-\notetpd\relax}

\let\oldtodo\todo

\makeatletter%
\DeclareDocumentCommand{\todo}{ O{} +g +d<> }{%
	\IfNoValueTF{#2}{\relax}{%
		\oldtodo[caption={#2},size=\footnotesize,#1]{\renewcommand{\baselinestretch}{1}\selectfont\sffamily#2\par}%
	}%
	\IfNoValueTF{#3}{\relax}{%
		\IfNoValueTF{#2}{
			\begin{note}%
				\begin{internallinenumbers}%
					\indent%
					#3%
				\end{internallinenumbers}%
			\end{note}%
		}{
			\vspace{-0\baselineskip}%
			\begin{note}%
				\begin{internallinenumbers}%
					\indent%
					#3%
				\end{internallinenumbers}%
			\end{note}%
		}%
	}%
}%
\makeatother
\usepackage{soul}

\newcommand{\removetodo}[2]{\todo{\textbf{delete:} ``#1'' #2}\hl{#1}}
\newcommand{\inserttodo}[1]{\todo[color=green!40]{\textbf{insert:} #1}}

\newcommand{\hltodoy}[2]{\todo[color=yellow!40]{#2}\hl{#1} }

\newcommand{\hltodo}[2]{\todo{#2}\hl{#1} }
\newcommand{\replacetodo}[2]{\todo[color=pink!40]{\textbf{replace with:}``#2'' }\hl{#1} }

\usepackage{marginnote}
\newcommand{\todol}[1][]{{%
		\let\marginpar\marginnote
		\reversemarginpar
		\renewcommand{\baselinestretch}{0.8}%
		\todo{#1}}}
	
\newcommand{\inserttodol}[1]{{%
		\let\marginpar\marginnote
		\reversemarginpar
		\renewcommand{\baselinestretch}{0.8}%
		\inserttodo{#1}}}
	
\newcommand{\removetodol}[2]{{%
		\let\marginpar\marginnote
		\reversemarginpar
		\renewcommand{\baselinestretch}{0.8}%
		\removetodo{#1}{#2}}}

\newcommand{\hltodol}[2]{{%
		\let\marginpar\marginnote
		\reversemarginpar
		\renewcommand{\baselinestretch}{0.8}%
		\hltodo{#1}{#2}}}

\newcommand{\replacetodol}[2]{{%
		\let\marginpar\marginnote
		\reversemarginpar
		\renewcommand{\baselinestretch}{0.8}%
		\replacetodo{#1}{#2}}}

\newcommand{\hltodoyl}[2]{{%
		\let\marginpar\marginnote
		\reversemarginpar
		\renewcommand{\baselinestretch}{0.8}%
		\hltodoy{#1}{#2}}}

\newtheorem{theorem}{Theorem}[section]

\newtheorem{lemma}[theorem]{Lemma}

\newtheorem{proposition}[theorem]{Proposition}

\newtheorem{corollary}[theorem]{Corollary}

\newtheorem{definition}{Definition}[section]

\def\bkE{{\rm I\kern-.17em E}}
\def\bk1{{\rm 1\kern-.17em l}}
\def\bkD{{\rm I\kern-.17em D}}
\def\bkR{{\rm I\kern-.17em R}}
\def\bkP{{\rm I\kern-.17em P}}

\def\bkZ{{\bf{Z}}}

\def\bkE{{\rm I\kern-.17em E}}
\def\bk1{{\rm 1\kern-.17em l}}
\def\bkD{{\rm I\kern-.17em D}}
\def\bkR{{\rm I\kern-.17em R}}
\def\bkP{{\rm I\kern-.17em P}}

\makeatletter
\newcommand{\pushright}[1]{\ifmeasuring@#1\else\omit\hfill$\displaystyle#1$\fi\ignorespaces}
\newcommand{\pushleft}[1]{\ifmeasuring@#1\else\omit$\displaystyle#1$\hfill\fi\ignorespaces}
\makeatother

\def\bkZ{{\bf{Z}}}
\def\b12{(\beta_1,\beta_2)}

\newenvironment{example}{{\noindent \bf Example}}{\hfill $\square$\hspace{-4.5pt}\vspace{6pt}}
\newcounter{example}
\renewcommand{\theexample}{\thesection.\arabic{example}}

\newcounter{remark}
\renewcommand{\theremark}{\thesection.\arabic{remark}}

\newenvironment{remark}{{\noindent \it Remark: }}{\hfill $\square$}

\def\Xscr{\mathcal{X}}
\def\Yscr{\mathcal{Y}}

\def\Ebb{\mathbb{E}}
\newlength{\noteWidth}
\long\def\notes#1{\ifinner
{\tiny #1}
\else
\marginpar{\parbox[t]{\noteWidth}{\raggedright\tiny #1}}
\fi\typeout{#1}}

 \def\notes#1{\typeout{read notes: #1}} 

\newcommand{\wi}[1]{\widehat{#1}}

\newcommand{\expec}[1]{\Ebb\Big[#1\Big]}

\newcommand{\curly}[1]{\left\{#1\right\}}
\newcommand{\round}[1]{\left(#1\right)}

\newcommand{\TypeK}{T_{\Yscr}^{K}}
\newcommand{\capacity}{\Xi(\mathscr{U})}
\newcommand{\Gs}{G_{\mathsf{s}}}
\newcommand{\Gsym}{G_{\mathsf{s}}^{\mathsf{Sym}}}
\newcommand{\Gc}{G_{\mathsf{c}}}
\newcommand{\ut}{\mathscr{U}}
\newcommand{\utsym}{\mathscr{U}^{\mathsf{Sym}}}

\newcommand{\best}{\mathscr{B}}
\newcommand{\pbest}{\mathscr{S}}

\newcommand{\ie}{i.e.\@\xspace} 
\newcommand{\eg}{e.g.\@\xspace} 

\newcommand{\Real}{\ensuremath{\mathbb{R}}}

\newcommand{\maximize}[1]{\displaystyle\maxim_{#1}}
\newcommand{\maxim}{\mathop{\hbox{\rm max}}}

\def\OPT{{\rm OPT}}

\def\rarr{\rightarrow}

\def\Ebb{\mathbb{E}}

\def\Pbb{{\mathbb{P}}}
\def\Nbb{{\mathbb{N}}}

\def\Ybb{{\mathbb{Y}}}
\def\Zbb{{\mathbb{Z}}}

\def\exp{\mathop{\hbox{\rm exp}}}

\def\spose#1{\hbox to 0pt{#1\hss}}
\def\sub#1{^{\null}_{#1}}
\def\text #1{\hbox{\quad#1\quad}}

\def\xhat{{\hat x}}

\def\nthinsp{\mskip -2   mu}

\def\superstar{^{\raise 0.5pt\hbox{$\nthinsp *$}}}
\def\SUPERSTAR{^{\raise 0.5pt\hbox{$*$}}}

\def\lamstarT {\lambda^{\raise 0.5pt\hbox{$\nthinsp *$}T}}

\def\Ascr{{\cal A}}

\def\Dscr{{\cal D}}

\def\Oscr{{\cal O}}
\def\Pscr{{\cal P}}
\def\Qscr{{\cal Q}}

\def\Rscr{{\cal R}}

\def\Cscr{{\cal C}}
\def\Zscr{{\cal Z}}
\def\Xscr{{\cal X}}
\def\Yscr{{\cal Y}}

\def\Lbar{\skew{4.3}\bar L}

\def\Ltilde{\widetilde L}

\def\stilde{\widetilde s}

\def\th{^{\rm th}}

\def\Ltilde{\widetilde {L}}

\def\xhat{\skew{2.8}\widehat x}

\def\ybar{\skew3\bar y}

\def\ytilde{\skew3\widetilde y}

\def\zbar{\skew{2.8}\bar z}

\def\ztilde{\skew{2.8}\widetilde z}

\def\supp{{\rm supp}}

\def\non{\nonumber}

\let\forallnew\forall
\renewcommand{\forall}{\forallnew\ }
\let\forall\forallnew

		\def\bkE{{\rm I\kern-.17em E}}
		\def\bk1{{\rm 1\kern-.17em l}}
		\def\bkD{{\rm I\kern-.17em D}}
		\def\bkR{{\rm I\kern-.17em R}}
		\def\bkP{{\rm I\kern-.17em P}}
		\def\bkY{{\bf \kern-.17em Y}}
		\def\bkZ{{\bf \kern-.17em Z}}
		\def\bkC{{\bf  \kern-.17em C}}

{\begin{list}{}%
         {\setlength{\leftmargin}{#1}}%
         \item[]%
}
{\end{list}}

		\def\bsp{\begin{split}}
		\def\beq{\begin{eqnarray}}
		\def\bal{\begin{align*}}
		\def\bc{\begin{center}}
		\def\be{\begin{enumerate}}
		\def\bi{\begin{itemize}}
		\def\bs{\begin{small}}
		\def\bS{\begin{slide}}
		\def\ec{\end{center}}
		\def\ee{\end{enumerate}}
		\def\ei{\end{itemize}}
		\def\es{\end{small}}
		\def\eS{\end{slide}}
		\def\eeq{\end{eqnarray}}
		\def\eal{\end{align*}}
		\def\esp{\end{split}}
		\def\qed{ \vrule height7.5pt width7.5pt depth0pt}  

	\def\cp2problem#1#2#3#4{\fbox
		 {\begin{tabular*}{0.9\textwidth}
			{@{}l@{\extracolsep{\fill}}l@{\extracolsep{6pt}}l@{\extracolsep{\fill}}c@{}}
				#1 & & $#4 $ 
			\end{tabular*}}}

		\def\bkE{{\rm I\kern-.17em E}}
		\def\bk1{{\rm 1\kern-.17em l}}
		\def\bkD{{\rm I\kern-.17em D}}
		\def\bkR{{\rm I\kern-.17em R}}
		\def\bkP{{\rm I\kern-.17em P}}
		
		\def\bkZ{{\bf{Z}}}

\newcommand {\beeq}[1]{\begin{equation}\label{#1}}
\newcommand {\eeeq}{\end{equation}}
\newcommand {\bea}{\begin{eqnarray}}
\newcommand {\eea}{\end{eqnarray}}

\def\texitem#1{\par\smallskip\noindent\hangindent 25pt
               \hbox to 25pt {\hss #1 ~}\ignorespaces}

\def\bsp{\begin{split}}
		\def\beq{\begin{eqnarray}}
		\def\bal{\begin{align*}}
		\def\bc{\begin{center}}
		\def\be{\begin{enumerate}}
		\def\bi{\begin{itemize}}
		\def\bs{\begin{small}}
		\def\bS{\begin{slide}}
		\def\ec{\end{center}}
		\def\ee{\end{enumerate}}
		\def\ei{\end{itemize}}
		\def\es{\end{small}}
		\def\eS{\end{slide}}
		\def\eeq{\end{eqnarray}}
		\def\eal{\end{align*}}
		\def\esp{\end{split}}
		\def\qed{ \vrule height7.5pt width7.5pt depth0pt}  

\usepackage{amsmath, amssymb, xspace}
\usepackage{epsfig}
\usepackage{longtable}
\usepackage{color}
\usepackage{mathrsfs}
\usepackage{subfig}
\newenvironment{proof}[1][]{{\noindent \emph {Proof} #1: }}{\hfill \qed \vspace{3pt}\\ }

\def\Cscr{{\cal C}}

\def\sub{\hbox{\rm s.t}}

	\def\maxproblemsmallsingcol#1#2#3#4{\fbox
		 {\begin{tabular*}{0.95\columnwidth}
			{@{}l@{\extracolsep{\fill}}l@{\extracolsep{-4pt}}l@{\extracolsep{\fill}}c@{}}
				#1 &  & $\maximize {#2}$ $#3$ & $ $ \\[4pt]
					 $\sub \ $  &   & $#4$ & $ $
			\end{tabular*}}
                    }

\usepackage[left=1in,right=1in,bottom=2.9cm,top=1.9cm,a4paper]{geometry}
\begin{document}

\begin{frontmatter}

\title{\textbf{Shannon meets Myerson: Information Extraction from a Strategic Sender}\tnoteref{titlenote}}
\tnotetext[titlenote]{The results in this paper were presented in part at IEEE International Conference on Acoustics, Speech and Signal Processing (ICASSP)~\cite{vora2021optimal} held virtually in June 2021, at the IEEE International Conference on Signal Processing and Communications (SPCOM)~\cite{vora2020zero} held virtually in July 2020 and  at the IEEE Conference on Decision and Control held virtually in December 2020~\cite{vora2020information}.}

\author{Anuj S. Vora, Ankur A. Kulkarni}
\address{Indian Institute of Technology Bombay, Mumbai-400076, India}
\address{anujvora@sc.iitb.ac.in, kulkarni.ankur@iitb.ac.in}




\begin{abstract}

We study a setting where a receiver must design a questionnaire to recover a sequence of symbols known to strategic sender, whose utility may not be incentive compatible. 
We allow the receiver the possibility of selecting the alternatives presented in the questionnaire, and thereby linking decisions across the components of the sequence. We show that, despite the strategic sender and the noise in the channel, the receiver can recover exponentially many sequences, but also that exponentially many sequences are unrecoverable even by the best strategy.  We define the growth rate of the number of recovered sequences as the  information extraction capacity. A generalization of the Shannon capacity, it characterizes the optimal amount of communication resources required. We derive bounds leading to an exact evaluation of the information extraction capacity in many cases. Our results form the building blocks of a novel, noncooperative regime of communication involving a strategic sender.

\end{abstract}

\begin{keyword}
Mechanism design, information theory, Stackelberg game,  questionnaires, screening
\end{keyword}

\end{frontmatter}


\section{Introduction}

Consider the following situation that arose during the early days of the Covid-19 pandemic. Travellers arrived at airports with varied travel histories and   
health inspectors had to screen these travellers based on responses to standardized questionnaires. Travellers arriving from unsafe locations were hesitant to reveal their true travel histories due to  inconvenience of quarantine protocols and stigma associated with the disease while those that arrived from safe locations wanted their true travel histories to be recorded. Some travellers had complex journeys where for some days they were at safe locations, while for other days they were at unsafe ones, and were perhaps inclined to selective misreporting. The health inspector's challenge was to design a questionnaire that recovered as many true travel histories as possible. 

The above situation can be posed as follows. A \textit{receiver}  (health inspector) wishes to   recover information privately known to the \textit{sender} (traveller) over a possibly imperfect communication medium. 
The private information of the sender, \ie, its type, is a sequence of $ n $ symbols (locations) each drawn from a finite set $ \Xscr $. A questionnaire is characterized by a  set of \textit{alternatives}, each drawn from $ \Xscr^n := \prod_{1}^{n} \Xscr$, and the sender is required to select one alternative as a function of its type, as its reported sequence. On viewing the sequence reported by the sender, the receiver applies a `decoding' function or an interpretation, mapping it to a decoded  sequence. 
The sender is a \textit{non-cooperative} agent and wishes to maximize a \textit{utility function} $ \ut_n $ that depends on the true sequence and the sequence decoded by the receiver. 
 This utility function dictates the sender's response to the questionnaire, and maximizing the utility may not align with the interests of the receiver. 
We assume that the $ n $-block utility function $ \ut_n $ takes the following form,
 \begin{equation}\label{eq:avg-util-defn}
 	\ut_{n}(\xhat,x)= \frac{1}{n} \sum_{i=1}^{n}\ut(\xhat_i,x_i),
 \end{equation}
 where $ x = (x_1,\hdots,x_n) \in \Xscr^n $ is the sender's type and 
 $ \xhat =(\xhat_1,\hdots,\xhat_n)\in \Xscr^n$ is the sequence of symbols decoded by receiver, and $ \ut : \Xscr \times \Xscr \rightarrow \Real $ is the single-letter utility of the sender. 
Depending on $ x$, for certain symbols, the sender may prefer that the receiver decodes the symbol incorrectly, whereas for other symbols it may want the receiver to know the truth. The receiver, on the other hand,  is interested in maximizing the number of true sequences $ x $ recovered correctly\footnote{For example, recall the tale of a mischievous boy (sender) who observes a source that can be two possible states: \textit{wolf} and \textit{no-wolf}. The villagers (receiver) want to know if there is indeed a wolf. But the boy derives a utility $ \ut (x',x)$ when the true state is $ x $ and the villagers decode it to be $ x' $. In the classical tale, when there is no wolf, the boy wants the villagers to think there is a wolf, \ie, $ \ut({\rm \textit{wolf,\ no-wolf}}) > \ut({\rm \textit{no-wolf,\ no-wolf}})$. But when there is a wolf, he wants them to infer that there is really a wolf, \ie, $ \ut({\rm \textit{wolf,\ wolf}}) > \ut({\rm \textit{no-wolf,\ wolf}})$.}. We ask the following question -- given the sender's tendency to  misreport its type and the possibly imperfect medium of communication, how must the receiver strategize to recover as much true information as possible? And what is the maximum \textit{amount} of information that the receiver can extract from the sender? We call this the problem of \textit{information extraction} from a \textit{strategic sender}.




Notice that  the receiver's options are rather limited. The receiver must publish a questionnaire before the sender's type is realized, and hence the questionnaire must be the same for all sender types. There is neither any scope for incentivizing truthful reporting by using transfers, nor do we assume incentive compatibility for the sender. How can the receiver then get any meaningful information from the sender? We allow the receiver the option to \textit{eliminate} certain alternatives from the questionnaire. If there are $ q $ possible symbols, an exhaustive (and naive) questionnaire would have all $ q^n $ sequences as possible alternatives for the sender to choose from. Instead, we allow the receiver to select only of a subset of these sequences and publish only those sequences in the subset as available alternatives\footnote{Our questionnaire requires the sender to select only one alternative.}.  This effectively constrains the signal space of communication between the sender and receiver and becomes a key tool for extracting non-trivial information from a strategic sender. Another tool with the receiver is the possibility of enforcing a linking of  decisions~\cite{jackson2007overcoming} across components of the sequence by suitably selecting the alternatives. The effect of this linking is that, not wanting to misreport  one leg of the history forces the traveller to truthfully report other legs too; this again allows the receiver to extract nontrivial information. 
Indeed, we find that although \eqref{eq:avg-util-defn} is additive, the optimal questionnaire for $ n $-length histories is \textit{not} a stacking of $ n $ questionnaires of $ 1$-length histories.

Our problem bears similarity to that of the implementation of a social choice function, as in Myerson's mechanism design setting \cite{myerson1997game}. Indeed, if modeled that way,  the social choice function for our problem would become the identity function and incentive compatibility\footnote{We use a strictly inequality in the definition in \eqref{eq:inc} because of the manner in which we break ties for the sender.} would ask that 
\begin{equation}\label{eq:inc}
	 \ut_n(x,x) >\ut_n(x',x) \qquad \forall \;x,x'\in \Xscr^n, x \neq x'.
\end{equation}
In this case, the sender's goal would coincide with the receiver's goal of recovering the truth, thereby reducing the problem to classical, non-strategic communication. Indeed in that extreme, our results do reduce to those from communication theory. Thus one may view our contribution as that of a \textit{non-cooperative theory of communication}, generalizing Shannon's communication theory to a setting more akin to Myerson's mechanism design.

\subsection{An example}
\label{sec:intro_eg}

The following example illustrates some important aspects of the setting. 
Suppose $ n = 1 $ and let $\Xscr = \{0,1,2,3\}$ be the set of locations. Let the utility of the sender $\ut : \Xscr  \times \Xscr \rarr \Real$ be given as
  \begin{align}
    \ut = \kbordermatrix{
    & 0 & 1 & 2 & 3 \\
   0 &0  & -1  & -2 & -2 \\
    1 & 0.5 &  0  & -1 & -2 \\
   2 & -1 &  0  &  0 & -1 \\
  3 & -1  & -1  &  0.5 &  0 \\
} \non
  \end{align}
Here $\ut(i,j)$, the entry in the  $ i\th $ row and $ j\th $ column, denotes the utility obtained by the sender when $i $ is the location decoded by the receiver and $j $ is the true location. From column $ 0 $, we can see  that $  \ut(0,0) < \ut(1,0)  $ and $ \ut(0,0) > \ut(3,0) =  \ut(2,0) $. Hence, the sender prefers that, when the true location is $0$, the receiver decodes it to be $1$, but not to $ 2,3 $. Similar observations for column $ 1 $ shows that the sender equally prefers $1$ and $2$ when the true location is $1$; column $ 2 $ shows that the sender prefers the location $ 3 $ when the true location is $ 2 $; the last column shows that the sender prefers $ 3 $ when the true location is $ 3 $.

 Suppose the receiver chooses a naive questionnaire $ \Cscr =\{0,1,2,3\} $ with all possible locations and announces to decode it with an identity function; \ie, the receiver's decoded location is equal to the location reported by the sender. For a true location $ x $, suppose the sender chooses an alternative $ s(x)$ from $ \Cscr $,  where $ s : \Xscr \rarr \Cscr $ is the response of the sender.  A location $ x $ is recovered by the receiver if $ x=s(x), $ and the set of locations recovered is $ \{x \in \Cscr : s(x)=x\} $. The goal of the receiver is to design $ \Cscr $ to maximize the worst case (over all best responses of the sender) size of this set.  It is easy to see that the sender has two best responses $s_1,s_2$, where
 \begin{align}
 s_1(i) = \left\{
 \begin{array}{c l}
 1& x = 0 \\
 1& x = 1 \\
 3& x = 2 \\
 3& x = 3
 \end{array}
 \right., \quad
 s_2(i) = \left\{
 \begin{array}{c l}
 1& x = 0 \\
 2& x = 1 \\
 3& x = 2 \\
 3& x = 3
 \end{array}
 \right.. \non
 \end{align}
%
We have  $ \{ x \in \Cscr : s_1(x) = x \}  = \{1,3\}$ and $ \{ x \in \Cscr : s_2(x) = x \}  = \{3\}$. Thus, in the worst case, only one location is recovered by the receiver.

However, the receiver can recover more than one location by cleverly choosing its questionnaire. The main difficulty encountered above is that when all locations are included as alternatives, it creates room for lying, whereby less truth is recovered. 
To counter this, suppose the receiver chooses a questionnaire $ \widetilde{\Cscr}  = \{0, 2 \}$, and an identity decoding function on $ \widetilde{\Cscr} $. 
The (unique) best response of the sender is now a strategy $\stilde$, where
    \begin{align}
    \stilde(i) = \left\{
    \begin{array}{c l}
      0& x = 0 \\
      2& x = 1 \\
      2& x = 2 \\
      2& x = 3
    \end{array}
         \right. \non. 
  \end{align}
The set of recovered locations is  $  \{ x \in \widetilde{\Cscr} : \stilde(x) = x \}  = \{ 0, 2\}$, thereby the receiver has recovered more truth. 
In $ \widetilde{\Cscr} $, the sender has only $0,2 $ as alternatives and $ 1,3 $ are left out. Given these choices, the sender is forced to report $0$ as $0$, unlike earlier where it was reporting $ 0 $ as $ 1 $ when $ 1 $ was available as an alternative. Similarly, it is forced to report $ 2 $ as $ 2 $.          We thus see that a more cunningly chosen questionnaire $ \widetilde{\Cscr} $  improves on the naive questionnaire $ \Cscr $ by forcing the sender to be truthful for the locations $ \{0,2\} $. 
  
  Such a questionnaire clearly leaves blind spots for the receiver -- there is no hope of recovering locations $ 1,3 $. But including $ 1$ in $ \widetilde{\Cscr} $ would jeopardize the recovery of $ 0 $, since the sender will report $ 0 $ as $ 1 $, and including $ 3 $ will preclude the recovery of $ 2 $, since the sender will report $ 2 $  as $ 3. $ It turns out that the receiver cannot recover more than two symbols in the worst case over the best responses of the sender.
    	One may ask if the receiver can recover any more locations, in the worst case, by choosing a different decoding function. We show later that this is not possible, and it suffices to consider the identity decoding function. In other words, $ \widetilde{\Cscr} $ is an optimal questionnaire.

  \subsubsection{Linking responses can increase recovered histories on average}

  Now let $ n = 2 $, \ie, the sender has $2$-length travel histories. The questionnaire will thus be composed of sequences from $ \Xscr^2 $. From the additive nature of the utility given in \eqref{eq:avg-util-defn}, it is easy to see (we also show this formally) that a questionnaire $ \widetilde{\Cscr}^2:= \widetilde{\Cscr} \times \widetilde{\Cscr} = \{ 00,02,20,22\} $ recovers all the travel histories in $ \widetilde{\Cscr}^2 $. Thus, we have that $ |\widetilde{\Cscr}^2|^{1/2}  = 2 $ and the receiver recovers the same amount of information \textit{per  unit length of the history}\footnote{Since the total number of travel histories for any $ n $ are $ \Xscr^n $, a natural way to compare the recovery of information from of questionnaires across different lengths of travel histories is to look at the $n$-th root of the size of the questionnaire.} as in $ \widetilde{\Cscr} $. However, can the receiver do better?

  Suppose the receiver declares a questionnaire $ \widehat{\Cscr} = \{ 00, 21, 02, 23, 30 \}$. Let $ 00 $ be the true sequence and consider another sequence $ 21 $ from $ \widehat{\Cscr} $. We have that
  \begin{align}
    \ut(21,00) = \frac{1}{2}\left( \ut(2,0) + \ut(1,0) \right) = \frac{1}{2}(-1+0.5) < 0,\non
  \end{align}
whereby the sender does not prefer to report $ 00 $ as $ 21. $
  Although the sender prefers the location $ 1 $ over the true location $ 0 $, the sender has to trade-off this benefit with the loss derived by reporting the location $ 0 $ as $ 2 $. Since the penalty from the latter is more than the incentive derived by misreporting, the sender prefers to report $ 00 $ over the sequence $ 21 $. This can be repeated for all pairs of sequences in $ \widehat{\Cscr} $ to show that whenever the sender's true sequence is one from $ \widehat{\Cscr} $, it prefers to report it truthfully. Thus, the receiver can recover all the histories from $ \widehat{\Cscr} $. Notice that $ |\widehat{\Cscr}|^{1/2} = 5^{1/2} > 2 $. In fact, this is the largest size of questionnaire for $n = 2 $. In other words, larger value of $ n $ opens up the possibility of creating linkages across legs of the journey whereby more information can be recovered than by mere stacking of the optimal questionnaires corresponding to 1-length histories.

  From this short example, a few observations are immediately evident. It is clear that questionnaires should not be designed innocently and the receiver has to \textit{strategize} in order to extract information from the sender. In general, the receiver may be able to extract only a subset of the information from the sender. Finally, the receiver can recover more information on average when the responses of the sender are linked. We formalize these observations in this paper.

\subsection{Main results}

 We formulate this problem as a leader-follower game with the receiver as the leader and the sender as the follower and analyze it based on the Stackelberg equilibrium. The receiver's strategy comprises of two parts -- it must decide the sequences to be retained as alternatives in the questionnaire, and it must decide a mapping that `interprets' the response of the sender by mapping it to a decoded sequence. 
 Our measure of information is the number of distinct sequences that the receiver recovers   in a Stackelberg equilibrium.   
 We find that there exists a Stackelberg equilibrium in which only this subset of the $ q^n $ sequences are retained as alternatives, and the receiver applies an identity decoding mapping on them\footnote{In other words, the receiver believes the sender, or takes the sender's responses at face value.}. The sequences are chosen in such a way that when \textit{restricted} to the chosen subset, the sender's utility $ \ut_n $ is incentive compatible (even though it may not be so on the full set of sequences $ \Xscr^n $). Thus, given the alternatives presented in the questionnaire, any sender whose true sequence is equal to one of the alternatives, has no incentive to misreport its true sequence. In the process, each of the sequences in the questionnaire is   recovered. 
 
We define the limiting value (with increasing length of sequences, henceforth termed as blocklength) of the exponent of the size of this set of   recovered sequences as the \textit{information extraction capacity  of the sender}. This capacity, denoted by $ \Xi (\ut)$,  is a function of the single-letterized utility function $ \ut$ of the sender. Operationally $ \capacity $ can be interpreted as the growth rate of the number of alternatives that are included in the optimal questionnaire. Equivalently, the capacity captures the rate at which the \textit{sheet of paper} on which the questionnaire is printed -- which is the communication medium between sender and receiver -- must grow with increasing length of the history. It is thus a measure of the optimal amount of communication resources required while communicating with a strategic sender.

 We find that $ \capacity $ is the limit of the $n$-th root of the independence number of a certain sequence of graphs whose structure is determined by $ \ut $. This is analogous to the notion of the Shannon capacity of a graph \cite{shannon1956zero}, a well-known and notoriously hard-to-compute quantity from communication theory. In fact, we show that the information extraction capacity \textit{generalizes} the Shannon capacity of a graph. 
 We show that the capacity lies between two important quantities,
\[ \Gamma(\ut) \leq \Xi(\ut) \leq \Theta(\Gsym). \]
These bounds are independent of the blocklength $ n $ and are a function of only the utility. The lower bound $ \Gamma(\ut) $ is the optimal value of an optimization problem defined over a set of permutation matrices on the alphabet $ \Xscr $. The upper bound is the Shannon capacity of a graph $ \Gsym $ induced by the symmetric part of the utility $ \ut $.

This result has a number of consequences. First, barring some corner cases, we have $  \Gamma(\ut)>1 $, whereby   the receiver can recover an \textit{exponential} number of sequences, regardless of any assumption of incentive compatibility on $ \ut $. Second, the upper bound shows that $ \capacity $ is, in general, strictly \textit{less} than $ q $, whereby the maximum number of sequences recovered is \textit{exponentially smaller} than the total number of sequences.
Third, for the utilities where the bounds match, the capacity  is exactly characterized. Examples include cases where $ \ut $ symmetric and the corresponding $ \Gsym $ is a perfect graph. We discuss more such examples later in the paper.
A well-known semi-definite program based upper bound on the Shannon capacity of a graph is the Lov{\'a}sz theta function introduced by Lov{\'a}sz in \cite{lovasz1979shannon}. This, along with the  lower  bound, together provide two computable bounds that can approximate the capacity when it is not exactly characterized. 
We also  derive a hierarchy of lower bounds as a function of the blocklength $ n $  by generalizing the bound $ \Gamma(\ut) $, which approach the capacity asymptotically as $ n $ grows large. This allows one to approximate the capacity arbitrarily closely from the left-hand side, albeit with increasing computational burden.

When the channel is noisy, the rate of information extraction is given by $\min \{\capacity, \Theta(\Gc)\},  $ where $ \Gc $ is the confusability graph of the channel. This result shows that as long as the zero-error capacity of the channel is greater than $ \capacity $, the receiver can extract the maximum possible information from the sender, and otherwise it is limited by the channel capacity. This result is analogous to the setting of joint source-channel coding in information theory, where the capacity of the channel is required to be larger than the entropy of the source to ensure reliable communication (\cite{cover2012elements}, Ch. 7). 

A general lesson in these results is that there are fundamental limitations to the operation of decentralized systems with self-interested agents. Problems where agents are compromised also occur in \textit{cyber-physical systems}, where a network of sensors connected over a communication medium is tasked with the operation of safety-critical applications such as nuclear power plants. A compromised sensor may be modeled as non-cooperative sender as in our setting. Our results show on the one hand how a receiver may strategize to obtain information from such agents, and on the other that there will usually be blind spots in the knowledge of the receiver, regardless of how it strategizes. Finally, when compromised sensors are present, improving the communication bandwidth only has a limited effect (it increases $ \Theta(\Gc) $ above); it may help the receiver hear the sender more clearly, but not more truth.






\subsubsection{Receiver's dilemma}

The key idea in the Stackelberg equilibrium strategy of the receiver is to limit the options of the sender for lying by selectively decoding only a subset of the sequences. If the receiver attempts to correctly recover a larger number of sequences by including more alternatives in the questionnaire, then the sender gets greater freedom to lie about its information, and the attempt of the receiver is counter-productive. On the other hand, if the receiver chooses too few sequences to include in the questionnaire, then the sender is compelled to speak truth given the limited choices, but the number of sequences recovered is less than optimal. The information extraction capacity is the growth rate of the ``optimal'' set that balances these two aspects.

Another subtlety here is that a strategic sender is distinct from a \textit{defunct} sender. A defunct sender sends arbitrarily corrupted messages, whereas a strategic sender's messages being motivated by its utility have an underlying structure. The optimal questionnaire exploits this structure to obtain nontrivial information from the sender.

Finally, our definition of questionnaire is one where the sender is required to ``select any one'' from a set of alternatives. More general questionnaires can be considered -- \eg, those where the sender can ``select all that apply'', or those with a ``none of the above" alternative -- but these are beyond our present scope. 
\subsubsection{Relation to cooperative information theory}

In the non-cooperative setting we consider, finite blocklength leader-follower games take the place of finite blocklength coding problems from the cooperative setting, whereas Stackelberg equilibria take the place of codes. 
We find that $ \Xi(\ut) $ plays a role loosely analogous to that played by the entropy of a source in cooperative communication; the utility $ \ut $ of the sender is akin to the probability mass function of the source. This analogy agrees with the rate of information extraction with a noisy channel, given by $ \min\{\Xi(\ut),\Theta(\Gc)\} $. Having said that, this analogy is loose because we are not concerned with a stochastic setting in this paper (recall we work in the zero-error regime). A much more complicated setting would arise when considering vanishing probability of error; our preliminary work on this line can be found in~\cite{vora2020achievable}. 

Another related viewpoint from information theory is the notion of large blocklength analysis. Information theory shows that one can, in general,  communicate more information on average by exploiting structure in \textit{sequences}, rather than in individual symbols, an idea commonly referred to as \textit{coding}. The linking of decisions we exploit in our work is the analogue of coding in our strategic setting. 
In information theory it is of interest to study how the codes scale with the blocklength and thereby help quantify the improvement in accuracy of communication with increase in channel capacity. Our notion of the information extraction capacity is an attempt at capture the similar requirements in a strategic setting.

When viewed from the communication standpoint, our channel input and output spaces are both equal to the space of source sequences. In the cooperative setting (in the noiseless case) this would trivially lead to recovery of all the source sequences. However, the same does not hold in our setting since the receiver chooses to selectively decode only a portion of the outputs in its optimal strategy. 
Thus, our results also quantify the \textit{optimal} amount of channel resources that are required for the receiver to extract information from the sender. In a sense, this marks a shift from the communication-theoretic concept of the \textit{capacity} of a channel to that of \textit{capacity utilization}, as something more relevant for the non-cooperative setting. 

\subsection{Related work}

There have been works on strategic communication (of various flavours) in the game theory community, but to the best of our knowledge ours is the first formal information-theoretic analysis of information extraction. The first model of strategic communication was introduced by Crawford and Sobel in \cite{crawford1982strategic}. They considered a sender and a receiver with misaligned objectives and formulated a simultaneous move game between the sender and receiver. They showed that any equilibrium involves the sender resorting to a \textit{quantization} strategy, where the sender reports only the interval in which its information lies. Some variants and generalizations were subsequently studied in (Battaglini \cite{battaglini2002multiple}, Sarita{\c{s}} \textit{et al.} \cite{saritacs2015multi}). These works considered the Nash equilibrium solution of the game. Strategic communication in control theory has been studied by Farokhi \textit{et al.} in \cite{farokhi2016estimation} and Sayin \textit{et al.} in \cite{sayin2019hierarchical}. The authors in \cite{farokhi2016estimation} studied a problem of static and dynamic estimation in the presence of strategic sensors as a game between the sensors and a receiver. The authors in \cite{sayin2019hierarchical} considered a dynamic signaling game between a strategic sender and a receiver. Strategic communication has been studied from the perspective of information theory by Akyol \textit{et al.} in \cite{akyol2015privacy, akyol2016information} where they studied a sender-receiver game and characterized equilibria satisfying a certain rate and distortion levels. In \cite{akyol2016information}, they also analyzed the effect of side information at the receiver. Strategic communication game is also studied in the information design setting \cite{kamenica2011bayesian, bergemann2019information}, also called the Bayesian persuasion problem, where the sender with superior information tries to influence the actions of the receiver. Information theoretic analysis of Bayesian persuasion problem was studied by Le-Treust and Tomala in \cite{letruest2019persuasion} where they derived an upper bound on the payoff achieved by the sender while communicating across a noisy channel. The works \cite{farokhi2016estimation}-\cite{letruest2019persuasion} have formulated the game with the sender as the leader.



 Our work differs from the above models as follows. We study the problem from the perspective of the receiver and hence we formulate the game with the receiver as the leader. We consider a model where the malicious behaviour of the sender is explicitly governed by a utility function. Our main contribution is the notion of information extraction capacity and the lower and upper bounds on this quantity. 

As mentioned earlier, our setting can be viewed as a problem of implementing a social choice function that may not be incentive compatible with the preferences of the agents. A related setting is studied by Jackson and Sonnenschein in \cite{jackson2007overcoming} where they demonstrated that these incentive constraints can be overcome by \textit{linking} independent copies of the decision making problem. They devised a mechanism and showed that as the number of linkages grow large, the mechanism implements the social choice function \textit{asymptotically}, \ie, the probability of the decisions on which the function cannot be implemented tends to zero. There are certain parallels between our setting and the setting considered in \cite{jackson2007overcoming}. For instance, the linking in \cite{jackson2007overcoming} is akin to the block structure of our setting and the implementation of the function is analogous to information recovery by the receiver. Viewed in this manner, our inquiry can be stated as follows -- how many decisions can be implemented by the social choice function exactly? Further, how does this number grow with the length of sequence of signals? Our setting is thus a \textit{zero-error} counterpart of the implementation problem.

In the context of information theory, our problem relates to the problem of coding in presence of mismatched criteria that has been studied extensively (see \cite{scarlett2020information} for a survey). The mismatch is in the encoding and decoding criteria and is to model the inaccurate or asymmetric information about the channel or to incorporate constraints on encoding or decoding. The optimal functions of the encoder and decoder in such cases are therefore chosen cooperatively, \ie, they are chosen with the common aim of achieving communication between the sender and the receiver. Thus, they do not capture the strategic nature of the problem. The problem thematically closest to our setting is the mismatched distortion problem studied by Lapidoth in \cite{lapidoth1997role}. In this problem, the distortion criteria of the receiver and the sender are mismatched, and the receiver aims to construct a codebook such that its own distortion is minimized. The author determines an upper bound on the distortion for a given rate of communication. A crucial difference between the setting of \cite{lapidoth1997role} and our setting is that in the former, the objective of the sender does not depend on the sequence decoded by the receiver. This fails to capture the strategic nature of our problem where the sender is indeed affected by the actions of the receiver is therefore trying to influence the outcome by misreporting.  

Our work significantly extends the results of \cite{vora2021optimal,vora2020zero}. In \cite{vora2021optimal}, we considered a situation where a health inspector designs a questionnaire to screen travellers. In \cite{vora2020zero} we discussed a special case of sender's utility and showed that the information extraction capacity is bounded above by the Shannon capacity of a certain graph. We studied related strategic communication problems in \cite{vora2020achievable, vora2020communicating}  where the receiver tried to achieve asymptotically vanishing probability of error.  In \cite{vora2020achievable}, we considered an unconstrained rate setting and determined the rates achievable for reliable communication. In \cite{vora2020communicating}, we considered a rate limited setting and we determined sufficient and necessary conditions for reliable communication.

 The paper is organized as follows. We formulate the problem in Section~\ref{sec:prob_form}. In Section~\ref{sec:equil_noiseless}, we determine the equilibrium of the Stackelberg game with the noiseless channel. In Section~\ref{sec:exist_capa}, we discuss the existence of information extraction capacity and in Section~\ref{sec:lower_bound} and ~\ref{sec:upper_bound}, we derive lower bounds and upper bounds respectively on the capacity. Finally, we analyze the Stackelberg game with the noisy channel in Section~\ref{sec:equil_noisy}. Section~\ref{sec:concl} concludes the paper.

\section{Problem Formulation}
\label{sec:prob_form}

\subsection{Notation}

Random variables are denoted with upper case letters $X, Y, Z$ and their instances are denoted as lower case letters $x, y, z$. Matrices are also denoted by uppercase letters. The space of scalar random variables is denoted by the calligraphic letters such as $\Xscr$ and the space of $ n $-length vector random variables is denoted as $\Xscr^n$. To unclutter notation, vector random variables $X, Y, Z$ and their instances $x, y, z$ will be denoted without the superscript $ n $.  The set of probability distributions on a space `$ \cdot $' is denoted as $\Pscr(\cdot)$. The empirical distribution of a sequence $ x \in \Xscr^n $ is denoted as $ P_x $ and is defined as $ P_x(i) = |\{ k : x_k = i\}|/n $. The joint empirical distribution of sequences $ (x,y) $ is defined similarly and is denoted as $ P_{x,y} $. A graph $G$ is denoted as $G = (V, E)$ where $V$ is the set of vertices and $E$ is the set of edges. When two vertices $x,y \in V$ are adjacent, we denote it either as $(x,y) \in E$ or as $x \sim y$. An independent set in $ G $ is a subset $ S $ of $ V $ such that no two vertices in $ S $ are adjacent. For a graph $G$, the size of the largest independent set is denoted as $\alpha(G)$. For a function or a random variable, we denote $\supp(\cdot)$ as its support set. For an optimization problem `$ \cdot $', we denote $\OPT(\cdot)$ as its optimal value. Unless specified, the $ \exp $ and $ \log $ are with respect to the base $ 2 $. For any $ n \in \Nbb $, we denote $ \{1,\hdots,n \} $ by $ [n] $.


\subsection{Model}
\label{sec:noiseless}

We present a model where the medium of communication is noiseless; it will generalized later to allow for noisy communication. Let the alphabet be $\Xscr = \{0,1,\hdots,q-1\}$, where $q \in \Nbb$ is the alphabet size. The sender observes a sequence $X = (X_1,\hdots, X_n) \in \Xscr^n $, where $X_i$ are independent and identically distributed symbols drawn from a known distribution\footnote{In the context of travellers, the term `observes' implies that the sender \textit{arrives} with a sequence of travel locations.}. The sender sends a message  $s_n(X) = Y \in \Xscr^n$, where $s_n : \Xscr^n \rarr \Xscr^n$. 
 The message is relayed perfectly to the receiver who decodes the message as $g_n(Y) = \wi{X}$, where $g_n : \Xscr^n \rarr \Xscr^n \cup \{\Delta\} $. Here $\Delta$ is an error symbol we introduce for convenience; we explain its meaning subsequently.
Let
\begin{align}
\Dscr(g_n,s_n) :=  \curly{x \in \Xscr^n \;|\; g_n \circ s_n(x) = x} \label{eq:decod-set-noiseless}
          \end{align}
          be the set of   recovered sequences when the receiver plays the strategy $g_n$ and the sender plays the strategy $s_n$. We also refer to the set of recovered sequences $ \Dscr(g_n,s_n) $ as the set of sequences \textit{correctly} decoded by the receiver.
          
 The receiver aims to maximize the size of the set $\Dscr(g_n,s_n)$ by choosing an appropriate strategy $g_n$. The sender, on the other hand, maximizes  $\ut_n(g_n \circ s_n(x),x)$ for each $ x $ by choosing an appropriate strategy $s_n$, where $\ut_n : (\Xscr^n \cup \{\Delta\}) \times \Xscr^n \rarr \Real$.  
The utility of the sender when $x$ is the true source sequence and $\wi{x}$ is the sequence decoded by the receiver is given by $ \ut_n(\wi{x},x) $ defined as \eqref{eq:avg-util-defn}.  Further, we assume   $\ut_n(\Delta,x) = -\infty$ for all $x \in \Xscr^n$ and $n \in \Nbb$. Thus $ \Delta $ is an outcome that is never preferred by the sender. 
We also assume that $ \ut(i,i)=0 $ for all $ i\in \Xscr $; this is without loss of generality as we explain below.

The operational meaning   of the above model in the context of the questionnaire design is as follows. The alphabet $ \Xscr $ is the set of locations and $ n $ is the length of the history. The message $ Y =s_n(X)$ is the alternative selected by a traveller with history $ X \in \Xscr^n $. $ g_n(Y) $ is the interpretation applied by the health inspector to the alternative $ Y $. Since $ \Delta $ is an outcome never preferred by the sender, if an alternative $ Y $ is mapped to $ \Delta $, it is equivalent to the alternative $ Y $ being not present in the questionnaire. Thus  $\Cscr_n :=\{y \in \Xscr^n \;|\; g_n(y) \neq \Delta \} $ is the set of alternatives presented to the travellers  and a traveller has to choose exactly one of the alternatives from this list. Thus if $ x $ is the true travel history of the traveller, then selecting an alternative $ y \in \Cscr_n $ amounts to setting $ s_n(x) = y $. The inspector then, is said to decode the response of the traveller to a travel history $ \xhat \in \Xscr^n $ if $ \xhat = g_n(y) $. The inspector recovers the travel history $ x $ if and only if $ g_n \circ s_n(x) = x $.

 We formulate this problem as a leader-follower game  with the receiver as the leader and the sender as the follower \cite{basar99dynamic}.
 \begin{definition}[Stackelberg equilibrium]\label{defn:stack-equi-noiseless}
The Stackelberg equilibrium strategy of the receiver is given as
\begin{align}
g_n^* \in \argmax_{g_n} \min_{s_n \in \best(g_n)} |\Dscr(g_n,s_n)|, \label{eq:rec-opt-stra-game-noiseless} 
\end{align}
where the set of best responses of the sender, $\best(g_n)$, is given as
\begin{align}
  \best(g_n) = {\Big\{} &s_n : \Xscr^n \rarr \Xscr^n \;|\; 
  \ut_n(g_n \circ s_n (x),x) \geq \ut_n(g_n \circ s_n' (x),x) 
  \quad \forall \; x \in \Xscr^n, \forall \; s_n' {\Big\}}. \label{eq:sen-opt-stra-game-noiseless}
\end{align}   
Any strategy $ s_n^* \in \best(g_n^*) $ is said to be a Stackelberg equilibrium strategy of the sender and the pair $ (g_n^*,s_n^*) $ is said to be a Stackelberg equilibrium.
 \end{definition}
It is easy to see that the set the set of best responses $ \best(g_n) $ is the same if $ \ut_n(x,x) $ is subtracted on both sides of the inequality in \eqref{eq:sen-opt-stra-game-noiseless}. 
 Thus, without loss of generality,  we assume  $ \ut(i,i) = 0 $ for all $ i \in \Xscr $.
 

In \eqref{eq:rec-opt-stra-game-noiseless}, we minimize over the set $\best(g_n)$ of the best responses of the sender because the sender may have multiple best responses and the receiver does not have control over the choice of the sender's specific best response strategy. We assume that the receiver chooses its strategy according to the worst-case over all such best responses and hence adopts a \textit{pessimistic} viewpoint. This is also the formulation of Stackelberg equilibrium adopted in standard sources such as \cite{basar99dynamic}.

Note that the problem of information extraction we study is distinct from the problem of \textit{information design} \cite{kamenica2011bayesian, bergemann2019information} or \textit{information disclosure} \cite{akyol2016information}. In these cases, the problem is studied from the perspective of the sender who designs the information that is observed by the receiver so as to achieve an outcome that favours the sender. Thereby, in such a setting it is suitable to formulate the problem with the sender as the leader. 
 On the other hand, in the case of screening of travellers, it is the role of the receiver to `ask' the agents about their information. Therefore, it is apt to study the problem with the receiver as the leader of the game. 

\section{Information Extraction from the Sender in Equilibrium}
 \label{sec:equil_noiseless}
 
In this section, we characterize the size of the largest set of sequences recovered by the receiver in a Stackelberg equilibrium and its growth rate with the blocklength $ n $.
In general for a fixed receiver strategy $ g_n $, the size $ |\Dscr(g_n,s_n) |$ of the set of recovered sequences could vary as the best response $ s_n $ varies over $ \best(g_n) $. We consider the smallest such size as our notion of the number of recovered sequences. To characterize the growth rate of this size, we define the \textit{rate} of information extraction as follows.
\begin{definition}[\textit{Rate of information extraction for a strategy}] \label{defn:rate-noiseless}
The number of   recovered sequences by a receiver strategy $ g_n $ is defined as $$ \min_{s_n \in \best(g_n)} |\Dscr(g_n,s_n)|. $$	
	For any strategy $ g_n $ of the receiver, the rate of information extraction is defined as
	\begin{align}
		R(g_n) = \min_{s_n \in \best(g_n)} |\Dscr(g_n,s_n)|^{1/n}. \non
	\end{align}
\end{definition}

Below we characterize the number of sequences recovered in a Stackelberg equilibrium (Definition~\ref{defn:stack-equi-noiseless}) in terms of a graph induced by the utility of the sender on the space of sequences $\Xscr^n$, called the \textit{sender graph}.
\begin{definition} [Sender graph] \label{defn:util-graph}
The sender graph, denoted as $ \Gs^n = (\Xscr^n,E_{\sf s}) $, is the graph where $ (x,y) \in E_{\sf s} $ if either
  \begin{align}
\ut_n(y,x) \geq  0   \;\; \mbox{or}\;\;    \ut_n(x,y) \geq  0. \non
  \end{align}
For $n = 1$, the graph $\Gs^1$ is denoted as $\Gs$ and referred to as the \textit{base graph}.
\end{definition}
Thus, two vertices $x$ and $y$ are adjacent in $\Gs$ if the sender has an incentive to report one sequence as the other. 

\begin{remark}
  Two single-letter utility functions inducing the same base graph on $\Xscr$ can induce two different sequences of sender graphs on $\Xscr^n$ for     $ n > 1$. Let $\Xscr = \{0,1,2\}$ and consider the utilities $\ut$ and $ \ut' $ defined  as
   \begin{align}
\ut = \left(\begin{array}{c c c}
 0&  -1 &-2 \\
 1 &0 &-1  \\
    -1 &0& 0 
            \end{array} \right), \quad
           \ut' = \left(\begin{array}{c c c}
  0 & -2.5 & -2.5 \\
  1 & 0 & -1 \\
  -1.5 & 0 & 0
\end{array} \right). \non 
   \end{align}
 It is easy to see that the graph $ \Gs $ and $ \Gs' $ induced on $\Xscr$ by $ \ut $ and $ \ut' $ respectively is a path $0 - 1 - 2$.
   
   Now take $n = 2$ and consider the graphs $\Gs^2$ and $\Gs'^2$ induced by $\ut_2$ and $\ut_2'$ respectively. It can be shown that $01$ and $10$ are adjacent in the graph $\Gs^2$, but are not adjacent in the graph $\Gs'^2$. Thus, although $\Gs$ and $\Gs'$ are same, $\Gs^2$ and $\Gs'^2$ are different. In short, the graphs $\{\Gs^n\}$ are defined by the utility $\ut$ rather than $\Gs$ alone.
\end{remark}

We show that only those sequences can be recovered   by the receiver that form an independent set in  $ \Gs^n $. 
\begin{lemma}\label{lem:dec-set-ind-set}
Let $n \in \Nbb$. Consider a sender with utility $\ut$ and let $\Gs^n$ be the corresponding sender graph. For any strategy $g_n$  define,
     \begin{align}
       \pbest(g_n) &= \argmin_{s_n \in \best(g_n)} |\Dscr(g_n,s_n)|. \non
     \end{align}
     Then, for all strategies $s_n \in  \pbest(g_n)$, $\Dscr(g_n,s_n)$ is an independent set in $\Gs^n$.
   \end{lemma}
   \begin{proof}
     See Appendix~\ref{appen:lem3.1}.
   \end{proof}

Thus, $ R(g_n) $ can be at most $ \alpha(\Gs^n)^{1/n} $ for any strategy $ g_n $ of the receiver.  The next theorem shows  that by choosing an appropriate strategy $g_n$, the receiver can recover any of the  largest independent sets of $\Gs^n$ and consequently achieve the rate $ \alpha(\Gs^n)^{1/n} $.
\begin{theorem} \label{thm:equil-dec-set-noiseless}
Let $n \in \Nbb$. Consider a sender with utility $\ut$ and let $\Gs^n$ be the corresponding sender graph. For all Stackelberg equilibrium strategies $g_n^*$ of the receiver,
  \begin{align}
R(g _n^*) = \alpha(\Gs^n)^{1/n}. \non
  \end{align}
\end{theorem}
\begin{proof}
  See Appendix~\ref{appen:thm3.2}.
\end{proof}

We show in the proof of the above theorem that it is sufficient for the receiver to choose a strategy $ g_n $ as
  \begin{align}
    g_n(x) = \left\{
    \begin{array}{c l}
      x  & \mbox{if} \; x \in I^n \\
      \Delta & \mbox{if} \; x \notin I^n
    \end{array}
\right., \label{eq:defn-gn-thm}
  \end{align}
where $ I^n $ is any largest independent set in $ \Gs^n. $
  Thus, the receiver decodes meaningfully only for sequences in $ I^n $. For the rest of the messages, the receiver maps them to $\Delta$. Operationally we can interpret this strategy as a questionnaire in which only the sequences from $ I^n $ are retained as alternatives and all other sequences are dropped. Any alternative selected by the sender is taken at face value by the receiver, \ie, the receiver applies an identity decoding function. Since $ I^n $ is an independent set in $ \Gs^n $ it follows from Definition~\ref{defn:util-graph} 
  that 
  \[ \ut_n(x,x) > \ut_n(y,x) \qquad \forall x,y \in I^n, x \neq y. \]
  In other words $ \ut_n $ is \textit{incentive compatible} with the identity function \textit{when restricted to the set} $ I^n.$ Thus truth-telling is the best response for senders whose type belongs to $ I^n$, whereby $ I^n $ is recovered by the receiver. For sequences that do not belong to $ I^n $, the sender selects the sequence from $ I^n $ that gives it the maximum utility from amongst sequences in $ I^n $. These latter sequences are not correctly recovered by the receiver. 
  
  One may wonder if including sequences from $ \Xscr^n \backslash I^n $ could lead to even more sequences recovered. This is false -- inclusion of even a single additional sequence from $ \Xscr^n \backslash I^n $ will lead to at least one sequence not being recovered. This is because $ I^n $ is a maximum independent set, whereby every sequence in $ x \in \Xscr^n \backslash I^n $ is adjacent to at least one sequence in $ y \in I^n, $ and hence  the sender has an incentive to either report $ x $ as $ y $ or $ y $ as $ x $. In other words, at most one of $ x $ and $ y $ can be recovered.
  Finally, note that when $ \ut_n $ is incentive compatible on $ \Xscr^n $, we have $ I^n = \Xscr^n $ and every sequence is recovered in the equilibrium. 
  

\section{Information Extraction Capacity of the Sender}
\label{sec:exist_capa}

Theorem~\ref{thm:equil-dec-set-noiseless} shows that the maximum information the receiver can extract from the sender is equal to $ \alpha(\Gs^n) $ and the maximum rate is given by the $\alpha(\Gs^n)^{1/n} $. The limiting value of this quantity as $ n\rightarrow \infty $ indicates a fundamental limit to the amount of information that is obtainable from such a strategic sender. We call this the \textit{information extraction capacity} of the sender.
 

\begin{definition}[\textit{Information extraction capacity of a sender}] \label{defn:inf-ext-cap}
	Consider a sender with utility $\ut$ and let $ \{\Gs^n\}_{ n\geq 1 } $ be the corresponding  sequence of sender graphs. The information extraction capacity of the sender is defined as
	\begin{align}
	\capacity = \lim_{n \rarr \infty} \alpha(\Gs^n)^{1/n}. \non
	\end{align}
\end{definition}
We show the existence of the limit in Definition~\ref{defn:inf-ext-cap} in Appendix~\ref{appen:thm4.1}.  
If  $ \capacity $ is greater than unity, the receiver can extract an \textit{exponentially} large number of sequences from the sender when the channel is noiseless, whereas if $ \capacity=1 $, then asymptotically only a vanishing fraction of sequences can be recovered. If $ \capacity =q$, it means that almost all sequences can be recovered and the number of sequences \textit{not} recovered is asymptotically a vanishing fraction of the total number of sequences. 

Considering that we make no assumption about incentive compatibility on $ \ut $, it is rather interesting to note that $ q>\capacity >1$, except in some corner cases. We show this and other properties of the information extraction capacity in the following sections.




\subsection{Information extraction capacity generalizes Shannon capacity}

Our definition of the information extraction capacity is inspired by the definition of the Shannon capacity of a graph as given in \cite{lovasz1979shannon}.
The Shannon capacity is given in terms of the strong product graph which is defined as follows. %
\begin{definition} \label{defn:strong-pdt}
  $1)$ \textit{Strong product}: Let $G_1 = (V_1,E_1)$ and $G_2 = (V_2,E_2)$ be two graphs. Then the strong product of the graphs $G_1, G_2$ is given by a graph $G = (V,E)$ where $V = V_1 \times V_2$. Further, two vertices $(x,x'), (y,y') \in V$, with $x,y \in V_1$ and $x',y' \in V_2$, are adjacent if and only if one of the following holds
  \begin{itemize}
  \item $x = y$ and $x' \sim y'$
  \item $x \sim  y$ and $x' = y'$
  \item $x \sim y$ and $x' \sim y'$
    \end{itemize}
    The strong product operation is denoted as $\boxtimes$ and the product graph $G$ is written as $G = G_1 \boxtimes G_2$.
    
$2)$ \textit{Strong product graph}: The strong product graph denoted as $G^{\boxtimes n}$ is the graph constructed by taking the $n$-fold strong product of the graph $G$, \ie,
\begin{align}
    G^{\boxtimes n} = \underbrace{G \boxtimes G \boxtimes \hdots \boxtimes G}_{n}. \non
\end{align}
\end{definition}

\begin{definition}[\textit{Shannon capacity}] \label{defn:shannon-cap}
	Let $G$ be any graph. The Shannon capacity of $G$ is defined as 
	\begin{align}
	\Theta(G) = \lim_{n \rarr \infty} \alpha(G^{\boxtimes n})^{1/n}, \non
	\end{align}
	where $G^{\boxtimes n}$ is the $n$-fold strong product given by Definition~\ref{defn:strong-pdt}.
\end{definition}

In \cite{shannon1956zero}, Shannon investigated the problem of computing the maximum number of messages that can be transmitted across a noisy channel such that the receiver can recover the messages with zero probability of error. He introduced the notion of the \textit{confusability graph} $G$ induced by this channel and showed that for any blocklength $n$, $\alpha(G^{\boxtimes n})$ is the maximum number of messages that can be communicated perfectly across the channel. The limit of the quantity $ \alpha(G^{\boxtimes n}))^{1/n} $ was termed as the \textit{zero-error capacity} of the channel, also known as the Shannon capacity of a graph $ G $.

The Shannon capacity of a graph $ \Theta(G) $ is an important quantity with applications in combinatorics and computer science as well. It is, however, found to be very hard to compute barring some simple cases. For instance, the capacity of the pentagon graph was unknown for about $ 20 $ years. Further, the capacity of a heptagon graph is still unknown. A computable semi-definite program-based upper bound was introduced by Lov{\'a}sz in \cite{lovasz1979shannon}, now known as the Lov{\'a}sz theta number. The introduction of this quantity was a significant development and is known to lie between the NP-hard clique and chromatic numbers of a graph. However, it is known that the Lov{\'a}sz theta number is not a tight upper bound. For more discussion on this subject of capacity of a graph and its variations, the reader is referred to \cite{korner1998zero} and the references therein.

Remarkably, we show that the information extraction capacity in fact generalizes the Shannon capacity of a graph.

  \begin{theorem}\label{thm:generalization-shannon-cap}
  Consider the graph $ G $ and let the adjacency matrix of the graph be denoted as  $ \Ascr $.  Define  a utility $ \ut : \Xscr \times \Xscr \rarr \Real  $ as 
  \begin{align}
\ut(i,j) =     \left\{\begin{array}{c l}
                        0  & \mbox{if} \;\; i = j  \mbox{ or } \Ascr(i,j) = 1 \\
                        -1 & \mbox{if} \;\; i \neq j \mbox{ and } \Ascr(i,j) = 0
    \end{array}\right.. \non
  \end{align}
    Then, 
\begin{align}
\capacity = \Theta(G). \non
\end{align}
\end{theorem}
\begin{proof}
  See Appendix~\ref{appen:thm4.2}.
\end{proof}
Thus, the information extraction capacity generalizes the notion of the Shannon capacity of a graph. This shows that computing the exact information extraction capacity would be hard in general and we only strive to obtain bounds on this capacity. Nevertheless, there are few classes of utilities for which the capacity is exactly characterized and we shall discuss them in the forthcoming sections.  We will also discuss cases where the capacities are equal even when the two graphs are not equal.

 \subsection{Lower Bounds on the Information Extraction Capacity} 
\label{sec:lower_bound}
 
In this section, we present a lower bound on the information extraction capacity in terms of the optimal value of an optimization problem. We also discuss the characteristics of the feasible region. We derive improvements on the lower bound via a generalization of the optimization problem. Further, we also show that the lower bound based on this generalization is asymptotically tight. 

We first define the optimization problem. Let $\Yscr \subseteq \Xscr$. Let $\Qscr = \{Q^{(0)},\hdots,Q^{(|\Qscr|)}\}$ be the set of all $|\Yscr| \times |\Yscr|$ permutation matrices and  $Q^{(0)} = \mathbf{I}$ be the identity matrix. For convenience, we assume that the permutation matrices are indexed by the symbols from $ \Yscr $. Consider the problem $\Oscr(\ut)$ as 
  \begin{align}
\maxproblemsmallsingcol{$\Oscr(\ut) :$}
{\Yscr \subseteq \Xscr}
{\displaystyle  |\Yscr| }
{\begin{array}{r@{\ }c@{\ }l}
   \sum_{i,j \in \Yscr} Q(i,j)\ut(i,j)&<& 0 \quad \forall \;   Q \in \Qscr \setminus \{\mathbf{I}\}.
	\end{array}}         \non
\end{align}
Let $ \Gamma(\ut)  = \OPT(\Oscr(\ut)) $.

\begin{theorem} \label{thm:cap-lower-bound}
  For a sender with utility $ \ut $,
  \begin{align}
    \capacity \geq \Gamma(\ut) \geq  |\Yscr| , \non
  \end{align}
for any $ \Yscr \subseteq\Xscr $ that is feasible for $ \Oscr(\ut).$
\end{theorem}
\begin{proof}
  See Appendix~\ref{appen:thm5.2}.
\end{proof}
Notice that the problem $ \Oscr(\ut) $ is independent of the blocklength $ n $  and thus the lower bound $ \Gamma(\ut) $ is a single-letter condition. Computing $ \capacity $ requires the sequence of $ \{\alpha(\Gs^n)\}_{n \geq 1} $ corresponding to the sequence of sender graphs $ \{\Gs^n\}_{n \geq 1} $ induced by $ \ut $. Moreover, computation of the independence number of a generic graph can be intractable. We also show via the proof of the above theorem that for a fixed $ \Yscr $, the problem $ \Oscr(\ut) $ can be solved as a linear program. Thus, in this light, the above theorem provides a computationally efficient lower bound. 

We now provide a characterization of the sets that are feasible for the problem $ \Oscr(\ut) $. 

\begin{proposition}[Characterization of feasible region of $ \Oscr(\ut) $]\label{prop:i1-i2-condi-for-Y}
  A set $ \Yscr \subseteq \Xscr $ is feasible for $ \Oscr(\ut) $ if and only if  any sequence of distinct symbols $ i_0,i_1 ,\hdots, i_{K-1} \in \Yscr $,
  \begin{align}
    \ut(i_1,i_0) + \ut(i_2,i_1) + \hdots + \ut(i_0,i_{K-1}) < 0. \label{eq:neg-weight-chain}
  \end{align}
\end{proposition}
\begin{proof}
  See Appendix~\ref{appen:prop5.5}
\end{proof}
In the above proposition, the expression in \eqref{eq:neg-weight-chain} 
 is (upon scaling) the utility obtained by the sender when the observed sequence is  $ (i_0i_1\hdots i_{K-1}) $ and the decoded sequence is $ (i_1i_2\hdots i_{K-1}i_0) $, \ie, the receiver decodes the symbol $ i_1 $ in place of $ i_0 $, and $i_2  $ in place of $ i_! $ and so on for $ i_2 ,\hdots, i_{K-1} $. Thus these symbols form a \textit{chain of lies} created by the sender represented as 
$i_0 \rarr i_1 \rarr \hdots \rarr i_{K-1} \rarr i_0.$
   The above theorem shows that a set $ \Yscr $ is feasible for $ \Oscr(\ut) $ if and only if the utility obtained is negative for all possible chains of lies that can be formed from the symbols of the set $ \Yscr $. We also discuss a computational approach to check for the feasibility of a set in Appendix~\ref{appen:comp_approach}.

An immediate corollary is as follows. 
\begin{corollary}\label{corr:ind-set-feas-for-O}
  If $ \Yscr $ is an independent set in $ \Gs $, then $ \Yscr $ is feasible for $ \Oscr(\ut) $. Hence $ \Gamma(\ut) \geq \alpha(\Gs) $.
\end{corollary}
This bound implies that if $ \Gs $ is not a complete graph, then $ \capacity>1 $. $ \Gs $ is complete if the sender is a \textit{pathological liar}: for every pair of symbols $ i,j \in \Xscr $, the sender either prefers $ i $ to be recovered as $ j $ or $ j $ to be recovered as $ i. $ Notice that there is a slight opening here too for the receiver to exploit -- the sender may prefer $ i $ to be recovered as $ j $, but not $ j $ to be recovered as $ i $. In this case, although $ \Gs $ is complete and $ \alpha(\Gs)=1 $, we find that one could still have $ \Gamma(\ut) >1$, whereby there is a nontrivial information extraction capacity to such a sender. In fact, we present a utility in Example \ref{eg:complete_graph_full_cap}  where  $\capacity = \Gamma(\ut)=q  $.  In the latter sections, we shall see that in general $ \Gamma(\ut) $ is greater than $ \alpha(\Gs) $. 

We now present a theorem that gives a lower bound on $ \capacity $. For that we require a few definitions.   
\begin{definition}[Cycle and Positive-edges cycle]\label{defn:pos-wt-cycle}
 Consider a set of distinct vertices $ \{i_0,i_1,\hdots,i_{K-1}\}$ from a graph. The vertices form a $K$-length cycle in the graph if two vertices $ i_l,i_m $ are adjacent whenever $ l = (m+1) {\rm mod} \; K $. Further, the cycle is a positive-edges cycle if for all $ m $,
  \begin{align}
    \ut(i_l,i_m) \geq 0 \quad \mbox{whenever}\;\; l = (m+1) \mbox{mod} \;K. \non %
  \end{align}
\end{definition}

Suppose there exists a set of vertices $ \Yscr $ that do not form such a cycle. The following theorem  gives a sufficient condition for such a set to be feasible for $ \Oscr(\ut) $ and the size of this set to be a lower bound on $ \capacity $.

\begin{theorem} \label{thm:suff-cond-for-pos-LP}
	Consider a sender with utility $\ut$ and let $\Gs$ be the corresponding sender graph. Suppose there exists a set $ \Yscr $ such that there is no positive-edges cycle in the sub-graph induced by $ \Yscr $ and 
	\begin{align}
	&\min_{i,j \in \Yscr: \ut(i,j) < 0} |\ut(i,j)| > (|\Yscr|-1)  \max_{i,j \in \Yscr: \ut(i,j) \geq 0} \ut(i,j). \label{eq:least-pos-gre-max-neg}
	\end{align}
	Then,  $ \capacity \geq |\Yscr| $. 
      \end{theorem}
      \begin{proof}
        See Appendix~\ref{appen:thm7.2}
      \end{proof}

      Theorem~\ref{thm:suff-cond-for-pos-LP} gives a sufficient condition in terms of the negative and positive values of $ \ut(i,j) $ for a set $ \Yscr $ to be feasible for $ \Oscr(\ut) $. 
Using this result, the following example shows that $\capacity = q $ even when the base sender graph $ \Gs $ is a complete graph.
\begin{examp} \label{eg:complete_graph_full_cap}
Let $\ut : \{0,1,2\} \times \{0,1,2\} \rarr \Real $ and consider the  following form of $ \ut $, 
	\begin{align}
	\ut = \round{   \begin{array}{c c c }
		0 & 1 & 1   \\
		-4 &  0 & 1  \\
		-4 &  -4 &  0  \\
                        \end{array}}. \non
	\end{align}
        The graph induced by the utility is a $3$-cycle graph and is given as 
        \begin{center}
          \begin{tikzpicture}
            \draw (-0.3,0) node {$1$};
            \draw (-0.3,1.5) node {$0$};
            \draw (1.8,0) node {$2$};
            \filldraw [black]
            (0,0) circle [radius=1pt]
            (0,1.5) circle [radius=1pt]
            (1.5,0) circle [radius=1pt];
            \draw (0,0) -- (0,1.5) (1.5,0) -- (0,0) (0,1.5) -- (1.5,0) ;
          \end{tikzpicture}
        \end{center}
        This follows since $ 	\ut(0,1), \ut(0,2), \ut(1,2) > 0$. It can be easily observed that there is no positive-edges cycle in the graph. This is because for chain $ i \rarr j \rarr i $,  either $ \ut(i,j) < 0 $ or $ \ut(j,i) < 0 $. Further, for all chains $ i \rarr j \rarr k  \rarr i $,  either $ \ut(j,i) < 0 $ or $ \ut(k,j) < 0 $ or $ \ut(i,k) < 0 $. Also, the largest weight of a chain in $ \Gs $  is $ -4 + 1 + 1 = -2 $.
        Finally, $\min_{i,j : \ut(i,j) < 0} |\ut(i,j)| = 4 > (3-1) \max_{i,j : \ut(i,j) \geq 0}\ut(i,j) = 2$. Thus, the conditions of Theorem~\ref{thm:suff-cond-for-pos-LP} are satisfied and hence $ \Xscr $ is feasible for $ \Oscr(\ut) $ and hence $ \capacity = q =3$. 
\end{examp}

We find this example to be quite surprising -- even though $ \Gs $ is a complete graph, the information extraction capacity is the maximum it can be. The main lesson to be drawn from it is that the \textit{magnitude} of the gains or losses from truth-telling or lying determine the amount of information that can be extracted; the base sender graph $ \Gs $ only considers the sign of these quantities.


We now present some results using the symmetric part of the utility.
\begin{definition}[Symmetric part of a utility]\label{defn:symm-part-util}
For a given utility $ \ut $, let $ \utsym $ be the symmetric part defined as
\begin{align}
  \utsym(i,j) = \frac{1}{2}(\ut(i,j) + \ut(j,i)). \non 
\end{align}
\end{definition}
We denote the sender graph induced by $ \utsym $ as $ \Gsym $.
The following corollary shows that all feasible sets of $ \Oscr(\ut) $ are also independent sets in the symmetric graph $ \Gsym $.
\begin{corollary}\label{coro:Y-is-ind-set}
If a set $ \Yscr $ is feasible for $ \Oscr(\ut) $, then $ \Yscr $ is an independent set in $ \Gsym $ and  $ \Gamma(\ut) \leq \alpha(\Gsym) $. Equality holds when $ \ut$ is symmetric.
\end{corollary}
\begin{proof}
  See Appendix~\ref{appen:coro5.7}.
\end{proof}

We now use a generalized form of $ \Oscr(\ut) $ to derive a series of lower bounds which approach the capacity.
Recall the optimization problem $ \Oscr(\ut) $. Let $\Yscr_n \subseteq \Xscr^n$. Let $\Qscr = \{Q^{(0)},\hdots,Q^{(|\Qscr|)}\}$, $Q^{(0)} = \mathbf{I}$ be the set of all $|\Yscr_n| \times |\Yscr_n|$ permutation matrices. For convenience, we assume that the permutation matrices are indexed by sequences from $ \Yscr_n $. Consider the problem $\Oscr(\ut_n)$ as
  \begin{align}
\maxproblemsmallsingcol{$\Oscr(\ut_n) :$}
{\Yscr_n \subseteq \Xscr^n}
{\displaystyle  |\Yscr_n| }
{\begin{array}{r@{\ }c@{\ }l}
   \sum_{x,z \in \Yscr_n} Q(x,z)\ut_n(x,z) \quad \forall \;   Q \in \Qscr \setminus \{ \mathbf{I}\}.
	\end{array}}         \non
\end{align}
Let $\Gamma(\ut_n) = \OPT(\Oscr(\ut_n)) $.  The following theorem is a generalization of Theorem~\ref{thm:cap-lower-bound}.



   \begin{theorem}[Hierarchy of bounds] \label{thm:limit_of_heirarchy}
     For any utility $ \ut $, $ \capacity \geq \Gamma(\ut_n)^{1/n} $ for all $ n $. Further, 
     \begin{align}
       \lim_{n \rarr \infty} \Gamma(\ut_n)^{1/n} = \capacity. \non
     \end{align}
   \end{theorem}
   \begin{proof}
  See Appendix~\ref{appen:thm5.12}.
\end{proof}
The above theorem suggests a way to approximate the capacity arbitrarily closely by taking higher values of $ n $. In fact, improved bounds on the Shannon capacity  are also obtained similarly by using the fact that $ \alpha(G^{\boxtimes n}) \geq \alpha(G)^n$ \cite{polak2019new}. 
\subsection{Upper Bounds on the Information Extraction Capacity} 
\label{sec:upper_bound}
 
In this section, we derive upper bounds on the information extraction capacity of the sender. We show that the Shannon capacity of the sender graph corresponding to the symmetric part of the utility is an upper bound on the capacity. We also discuss a class of utilities  where the upper bound is given by the Shannon capacity of the sender graph itself.

First, consider the following upper bound for any utility $ \ut $.
   \begin{lemma} \label{lem:symm-util-upper-bnd}
     Consider a utility $ \ut $ and let $ \utsym $ be its symmetric part.  Then,
     \begin{align}
       \capacity \leq \Xi(\utsym). \non
     \end{align}
   \end{lemma}
   \begin{proof}
     See Appendix~\ref{appen:lem6.1}
   \end{proof}
Thus the information extraction capacity of a sender is no greater than that of a sender whose utility is equal to the symmetric part of the former's utility. In other words, ignoring the skew-symmetric part of the utility leads to an increase in the information extraction capacity. Intuitively this is because in the skew-symmetric part of the utility for each pair of symbols $ i,j \in \Xscr $, we have that either $  i$  is  preferred to be recovered as $ j $ or $ j $ is preferred to be recovered as $ i. $ Thus the skew-symmetric part has capacity $ 1 $.

\begin{theorem}\label{thm:upper-bnd-Shannon-cap-symm}
  Consider a sender with a utility $ \ut $ and let $ \utsym $ be its symmetric part. Let $ \Gsym $ be the sender graph corresponding to $ \utsym $. Then,
  \begin{align}
    \capacity \leq \Theta(\Gsym). \non
  \end{align}
\end{theorem}
  \begin{proof}
     See Appendix~\ref{appen:thm6.2}
   \end{proof}

As with the lower bound, the above theorem provides an upper bound that is independent of the blocklength $ n $ and depends only on the single letter utility $ \ut $. We can observe that the symmetric part of the utility plays a recurring role in the characterization of the bounds on the capacity. Recall that in Corollary~\ref{coro:Y-is-ind-set}, we derived a characterization of the feasible sets of the problem $ \Oscr(\ut) $ in terms of the symmetric sender graph and showed that $ \Gamma(\ut) \leq \alpha(\Gsym) $. 

We now present another class of utilities where the capacity is bounded above by the Shannon capacity of the sender graph.
  \begin{theorem}\label{thm:posi-diff-greater}
Consider a sender with utility  $\ut$ given as 
    \begin{align}
      \ut(i,j) = \left\{\begin{array}{c l}
                            a  & \mbox{if} \;\; \ut(i,j) \geq 0 \\
                            -b  & \mbox{if} \;\; \ut(i,j) < 0
      \end{array}\right., \non
    \end{align}
    where $a,b > 0$ and $a \geq b$. Let $ \Gs $ be the corresponding sender graph.  Then,
    \begin{align}
\capacity \leq \Theta(\Gs). \non
    \end{align}
  \end{theorem}
  \begin{proof}
     See Appendix~\ref{appen:thm6.4}
   \end{proof}
The utility in the above theorem is such that the \textit{incentive} for lying is greater than the \textit{penalty} for lying. Intuitively, a sender with utility as above has higher tendency to lie about its information, since it can offset its penalty of lying by gaining appropriate incentive. Thus the information extraction capacity for such a sender is in general strictly less than $ q. $
  
  The above theorem characterized the utilities for which the capacity is bounded above by $ \Theta(\Gs) $. In the following section, we will discuss cases where the capacity is exactly equal to $ \Theta(\Gs) $.

\subsection{Exact evaluation of  $ \capacity $}

In the earlier section, we discussed  upper bounds  on the information extraction capacity. A natural question that follows is that under what conditions the capacity is exactly characterized? In this section, we mention few cases of utility where the information extraction capacity is equal to the Shannon capacity of the induced sender graph.

\begin{theorem}\label{thm:exact-charac-of-capa}
  Consider a utility $ \ut $ and let $ \Gs $ be the corresponding sender graph. Then,
       \begin{align}
\capacity = \Theta(\Gs)  = \alpha(\Gs) \non
     \end{align}
if any of the following hold:
  \begin{enumerate}
     \item $ \ut $ is symmetric and $ \Gs $ is a perfect graph.
       
\item	$ \ut $ is of the form given in  Theorem~\ref{thm:posi-diff-greater} and  $ \Gs $ is a perfect graph.
\end{enumerate}
\end{theorem}
\begin{proof}
     See Appendix~\ref{appen:thm6.5}
   \end{proof}

The following example demonstrates how the lower bounds and the upper bounds can be used to exactly compute the capacity. 
\begin{examp}\label{eg:u-usym-cap-equal}
        Consider a $ \ut $ on $ \Xscr = \{ 0,1,2,3,4\} $ as
  \begin{align}
    \ut(i,j) = \left(
    \begin{array}{ c c c c c }
      0  &  -1 & -1 & -1 & -1  \\
       1 &  0 & -1  & -1 & -1  \\
      -1 & 1  & 0  & -1  & -1  \\
      -1 & -1 &  1 & 0   & -1   \\
      1 & -1 & -1 & 1    &  0
    \end{array}\right). \non
  \end{align}
  It can be observed that the base graph $ \Gs $ is a pentagon and is given as     \begin{center}
    \begin{tikzpicture}
      \filldraw [white]
      (0,0.6) circle [radius=1pt];
      \draw (0,2.3) node {$0$};
      \draw (-1.3,1) node {$1$};
      \draw (1.3,1) node {$4$};
      \draw (-1,-0.38) node {$2$};
      \draw (1,-0.38) node {$3$};
      \filldraw [black]
      (0, 2) circle [radius=1pt]
      (1, 1) circle [radius=1pt]
      (-1, 1) circle [radius=1pt]
      (0.707, -0.38) circle [radius=1pt]
      (-0.707, -0.38) circle [radius=1pt];
      \draw  (0, 2) -- (1,1) (0, 2) -- (-1,1)   (-1,1) -- (-0.707,-0.38)    (1, 1) -- (0.707,-0.38) (-0.707,-0.38) --  (0.707,-0.38);
    \end{tikzpicture}
  \end{center}
 This is because  $ i,j $ are adjacent if and only if $ |j-i| \mbox{mod}\; 5 = 1 $.  We show that $ \capacity = \sqrt{5} $.

 Notice that the graph does not  contain a positive-edges cycle since $ \ut(0,4) = -1 $. Further, the graph $ \Gsym $ induced by the symmetric part $ \utsym $ is also a pentagon graph since $ \utsym(i,j) = (\ut(i,j)+\ut(j,i))/2 = 0 $ if and only if  $ i, j $ are adjacent in $ \Gs $ or $ i = j $. For all other $ i,j $ not adjacent in $ \Gsym $, $ \utsym(i,j) < 0 $. Thus, from Theorem~\ref{thm:generalization-shannon-cap}, it follows that  $ \Xi(\utsym) = \Theta(\Gsym) = \sqrt{5} $. 

  We now compute $ \capacity $. From the problem  $ \Oscr(\ut) $, we only get a lower bound on $ \capacity $  since $ \Gamma(\ut) = 2 $. However, consider $ \Oscr(\ut_2) $ and the set $ \{00,12,24,31,43\} $. Let $ x,y $ be a distinct pair of sequences from $ \{00,12,24,31,43\} $. It follows that either $ \ut(x_k,y_k) = -1 $ or  $ \ut(y_k,x_k) = -1 $ and hence, $ \ut_2(x,y) \leq 0 $. Furthermore, $ \ut_2(x,y) = 0 $ if and only if $ x_1 = (y_1 + 1) \mbox{mod} \;5 $ and $ x_2 = (y_2 + 2) \mbox{mod} \;5  $. From this it is easy to observe that all sets of distinct sequences from the set $ \{00,12,24,31,43\} $  form an inequality as \eqref{eq:neg-weight-chain} and using a result analogous to Proposition~\ref{prop:i1-i2-condi-for-Y} for $2$-length sequences, we get that $ \{00,12,24,31,43\} $ is feasible for $ \Oscr(\ut_2) $  and hence
  \begin{align}
\capacity \geq \Gamma(\ut_2)^{1/2} = \sqrt{5}. \non
  \end{align}
The upper and lower bounds together give that $ \capacity = \sqrt{5} $.
\end{examp}

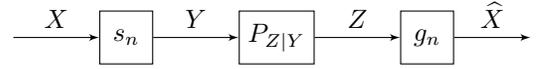
 \section{Information Extraction over a Noisy Channel}
 \label{sec:equil_noisy}
 
 We now discuss the case where the sender and receiver communicate via a noisy channel. We determine the maximum number of sequences that can be recovered by the receiver in any equilibrium of the game. We present a notion of the asymptotic rate of information extraction and we show that it is equal to the minimum of the information extraction capacity of the sender and the zero-error capacity of the noisy channel.

 \subsection{Model with a noisy channel}


Consider now a setting where the sender and receiver communicate via a noisy channel.  As earlier, the sender observes a sequence $X \in \Xscr^n$ and encodes it as $s_n(X) = Y$, where $s_n : \Xscr^n \rarr \Xscr^n$. However, the message is now transmitted to the receiver via a discrete memoryless channel which generates an output $Z \in \Xscr^n$ according to the distribution $P_{Z|Y}$ defined as
\begin{align}
  P_{Z|Y}(z|y) = \prod_{i=1}^nP_{\Zbb|\Ybb}(z_i|y_i), \label{eq:dmc-defn}
\end{align}
where $P_{\Zbb|\Ybb}(\cdot|\cdot) \in \Pscr(\Xscr|\Xscr)$. The  output  is decoded by the receiver  as $g_n(Z) = \wi{X}$, where $g_n : \Xscr^n \rarr \Xscr^n \cup \{\Delta\}$. Here $Z$ is distributed according to $P_{Z|Y}(\cdot|s_n(x))$, when $s_n(x)$ is the input to the channel.

Generalizing our earlier notation, let 
\begin{align} 
\Dscr(g_n,s_n) := \curly{x \in \Xscr^n \;|\; \Pbb(\wi{X} = x | X  = x ) = 1}, \label{eq:decod-set-noisy}
\end{align}
be the set of   recovered sequences when the receiver plays the strategy $g_n$ and the sender plays the strategy $s_n$. The receiver tries to maximize the size of this set by choosing a strategy $g_n$. 
The sender on the other hand chooses a strategy $s_n$ to maximize the expected utility 
\begin{align}
  \expec{\ut_n(\wi{X},x)} = \sum_{z \in \Xscr^n}   P_{Z|Y}(z|s_n(x))\ut_n(g(z),x) \non
\end{align}
for every $x \in \Xscr^n$. The utility $\ut_n$ is as given in \eqref{eq:avg-util-defn}.

As in the noiseless model, we pose the problem as a Stackelberg game.
\begin{definition}[Stackelberg equilibrium]\label{defn:stack-equi-noisy}
The optimal strategy of the receiver is given as
\begin{align}
  g_n^* \in \argmax_{g_n} \min_{s_n \in \best(g_n)}|\Dscr(g_n,s_n)|. \label{eq:rec-opt-stra-game-noisy}
\end{align}
 The set of best responses of the sender $\best(g_n)$ is determined as
\begin{align}
  \best(g_n) = \Big\{ &s_n : \Xscr^n \rarr \Xscr^n \;|\;
  \expec{\ut_n(g_n(Z),x)} \geq \expec{\ut_n(g_n(Z'),x)} 
   \quad \forall \; x \in \Xscr^n, \forall \; s_n' \vphantom{\expec{\ut_n(g_n(Z),x)}}\Big\}, \label{eq:sen-opt-stra-game-noisy}
\end{align}
where $Z$ is distributed according to $P_{Z|Y}(\cdot|s_n(x))$ and $Z'$ is distributed according to $P_{Z|Y}(\cdot|s_n'(x))$.  
\end{definition}
Notice that as in the noiseless channel model, we adopt a \textit{pessimistic} formulation for the receiver.

Thus, the set of best responses of the sender for a strategy $g_n$ of the receiver is a collection of strategies, $s_n$, such that for all sequences  $x$, the expected utility with respect to the distribution $P_{Z|Y}(\cdot|s_n(x))$ is the highest that can be obtained by the sender.

\subsection{Stackelberg equilibrium of the game}


Recall the definition of the recovered set $\Dscr(g_n,s_n)$ from \eqref{eq:decod-set-noisy}, when the receiver plays $g_n$ and the sender plays $s_n$. As in Definition~\ref{defn:rate-noiseless}, the number of sequences recovered by a strategy $ g_n $ of the receiver is defined as 
\[ \min_{s_n \in \best(g_n)} |\Dscr(g_n,s_n)|, \]
where $ \best(g_n) $ is as given in \eqref{eq:sen-opt-stra-game-noisy}.

We now recall the definition of a graph induced by the channel, called as the confusability graph.
\begin{definition} [Confusability graph] \label{defn:confus-graph}
The confusability graph of a channel $ P_{Z|Y} $, denoted as $\Gc^n = (\Xscr^n,E_{\mathsf{c}})$, is the graph where $(y,y') \in E_{\mathsf{c}}$ if there exists an output $z \in \Xscr^n$ such that
  \begin{align}
    P_{Z|Y}(z|y)P_{Z|Y}(z|y') > 0. \non
  \end{align}
  For $n = 1$, the graph $\Gc^1$ is denoted as $\Gc$.
\end{definition}
Thus, two inputs to the channel are adjacent in the confusability graph if they have a common output. The confusability graph was first introduced by Shannon in \cite{shannon1956zero}.

In the following theorem, we show that the receiver recovers $ \min\curly{\alpha(\Gs^n),\alpha(\Gc^n)} $ number of sequences in an equilibrium. The idea of the proof is as follows. We consider the independent sets in $\Gs^n$ and $\Gc^n$ of the size $\min\{ \alpha(\Gs^n),\alpha(\Gc^n)\}$ denoted as $I_{\mathsf{s}}^n$ and $I_{\mathsf{c}}^n$ respectively.  For every channel input sequence in $I_{\mathsf{c}}^n$, the receiver maps the corresponding output set to a unique sequence in $I_{\mathsf{s}}^n$. This establishes a one-to-one correspondence between the sequences in $I_{\mathsf{s}}^n$ and $I_{\mathsf{c}}^n$. Since $I_{\mathsf{s}}^n$ is an independent set in $\Gs^n$, the sender complies with the receiver and maps the sequences to their respective input sequences in $I_{\mathsf{c}}^n$. The receiver is thus able to recover the set $I_{\mathsf{s}}^n$ of size $\min\{ \alpha(\Gs^n),\alpha(\Gc^n)\}$. 

\begin{theorem} \label{thm:stac-eq-recov-ind-no}
Let $n \in \Nbb$. Consider a sender with utility $\ut$ and let $\Gs^n$ be the corresponding sender graph. Let $\Gc^n$ be the confusability graph of the channel $P_{Z|Y}$. For all Stackelberg equilibrium strategies $g_n^*$ of the receiver,
  \begin{align}
  \min_{s_n \in \best(g_n^*)}|\Dscr(g_n^*,s_n)| = \min\{ \alpha(\Gs^n),\alpha(\Gc^n)\}. \non  
  \end{align}
\end{theorem}
\begin{proof}
  See Appendix~\ref{appen:thm8.1}
\end{proof}
\subsection{Asymptotic rate of information extraction}

\begin{definition}[\textit{Asymptotic rate of information extraction}] \label{defn:max-rate}
	Let $\{g_n^*\}_{n \geq 1}$ be a sequence of Stackelberg equilibrium strategies for the receiver and let $ \{R(g_n^*)\}_{n \geq 1} $ be the corresponding sequence of rate of information extraction. The asymptotic rate of information extraction, denoted by $\Rscr$, is given as 
	\begin{align}
		\Rscr = \lim \sup_{n} R(g_n^*). \non
	\end{align}
\end{definition}
When the channel is noiseless, this quantity is equal to the information extraction capacity of the sender.  
We now use the results derived in the last section to determine the asymptotic rate of information extraction $\Rscr$ for the noisy channel case. 
\begin{theorem} \label{thm:limit-D-min-of-capa}
Consider a sender with utility $\ut$ and let $\Gs$ be the corresponding sender graph. Let $\Gc$ be the confusability graph of the channel $P_{Z|Y}$. Then the asymptotic rate of information extraction is given as 
  \begin{align}
    \Rscr = \min\{\capacity, \Theta(\Gc)\}. \non
  \end{align}
\end{theorem}
\begin{proof}
  See Appendix~\ref{appen:thm8.2}
\end{proof}
The above result states that given a sender with information extraction capacity $\capacity$, the zero-error capacity of the channel should be at least this number in order to extract maximum possible information from the sender. Alternatively, given the channel, the asymptotic rate of information extraction from any sender is bounded by the zero-error capacity of the channel. As long as both quantities are greater than unity, the receiver can extract exponentially large number of sequences from the sender.

Using this result and the bounds from earlier sections,  the following results follow.
\begin{theorem}
  Consider a sender with utility $ \ut $ and let $ \Gs $ be the corresponding sender graph. Let $ \utsym $ be its symmetric part and let $ \Gsym $ be the corresponding sender graph. Let $ \Gc $ be the confusability graph of the channel $P_{Z|Y}$.
  \begin{enumerate}
  \item If $ \Theta(\Gsym) \leq \Theta(\Gc)$, then 
  \begin{align}
  \Rscr = \capacity. \non  
  \end{align}
  \item If $\Gamma(\ut) \geq \Theta(\Gc)$, then
\begin{align}
\Rscr = \Theta(\Gc). \non  
\end{align}
\end{enumerate}
\end{theorem}
\begin{proof}
The proofs follow by using Theorem~\ref{thm:limit-D-min-of-capa} along with Theorem~\ref{thm:upper-bnd-Shannon-cap-symm} for part 1) and Theorem~\ref{thm:cap-lower-bound} for part 2) respectively.%
\end{proof}
The above result states that if the Shannon capacity of the symmetric sender graph is less than the zero-error capacity of the channel, then the asymptotic rate of information extraction is simply the information extraction capacity of the sender. The second part of the theorem states that the receiver can recover an exponential number of sequences even when the information extraction capacity of the sender is less than the zero-error capacity of the channel, provided $\Theta(\Gc) > 1$.

This concludes our analysis on this topic of information extraction from a strategic sender. We have seen that the strategic setting demands a new line of analysis, that uses in part the traditional tools of information theory, but is rooted in concepts of game theory. It also leads to new concepts. Our main take away is that the \textit{information extraction capacity of the sender}, a concept we defined and introduced in this paper, appears to be a fundamental quantity. It plays a role loosely analogous to that of the entropy of a source, characterizing the extent of information the sender can provide (or can be extracted from it).  Future research will reveal the extent to which this analogy holds.


 \section{Conclusion}
\label{sec:concl}

To conclude, inspired by the problem of screening of travellers with questionnaires, we considered a framework where a receiver attempted to extract information from a strategic sender. This setting was posed as a non-cooperative communication problem where the receiver (a health inspector) wishes to recover information from a misreporting sender (traveller) with zero probability of error. We considered a receiver-centric viewpoint and posed the problem as a leader-follower game with the receiver as the leader and sender as the follower. We formulated two instances of the game, with a noiseless channel, and with a noisy channel. We showed that even in the presence of the noisy channel, the receiver can extract an exponential number of sequences. To achieve this, the optimal choice of strategy for the receiver is to play a selective decoding strategy that decodes meaningfully only for a subset of sequences and deliberately induces an error on the rest of the sequences. The sequences are chosen such that the sender does not have an incentive to misreport any sequence as other, whereby, it tells the truth. In the context of designing questionnaires, this corresponds to the size of the optimal questionnaire that recovers maximum number of travel histories.

Our analysis led to new concepts: the rate of information extraction and the information extraction capacity of the sender. We showed that the maximum rate of information extraction is equal to the information extraction capacity of the sender in the noiseless channel case. In the presence of the noisy channel, the receiver can still extract information with this rate, provided the zero-error capacity of the channel is larger than the information extraction capacity of the sender. We derived single-letter lower bounds and upper bounds. The lower bound is the optimal value of  an optimization problem over permutation matrices. The upper bound is the Shannon capacity of the sender graph corresponding to the symmetric part of the utility.  The information extraction capacity characterizes the fundamental limit to the amount of information that can be recovered with questionnaires.

\section*{Appendix}
\appendices
\section{Preliminaries}

  \subsection{Proof of Lemma~\ref{lem:dec-set-ind-set}}
  \label{appen:lem3.1}
  \begin{proof}
    For strategies $g_n$ of the receiver  such that $ \min_{s_n \in \pbest(g_n)}|\Dscr(g_n,s_n)| \leq 1$, the claim trivially holds. Let $g_n$ be such that $|\Dscr(g_n,s_n)| \geq 2$ for all strategies $s_n \in \pbest(g_n)$.  We prove the claim by contradiction.
    
    Suppose for some strategy $s_n \in  \pbest(g_n)$, the set $\Dscr(g_n,s_n)$ is not an independent set in $\Gs^n$. 
    Thus, there exists distinct sequences $\bar{x}, \wi{x} \in \Dscr(g_n,s_n)$ such that $\ut_n(\bar{x},\bar{x}) \leq \ut_n(\wi{x},\bar{x})$.  Using this,  define a strategy $\bar{s}_n$ as 
    \begin{align}
      \bar{s}_n(x) = \left\{
      \begin{array}{c l}
        s_n(x)& \forall \; x \neq \bar{x} \\
        s_n(\wi{x}) & \mbox{for}\; x = \bar{x}
      \end{array}\right..
    \end{align}
    Observe that $\bar{s}_n$ is also a best response since
    \begin{align}
      \ut_n(g_n \circ \bar{s}_n(x),x) =        \ut_n(g_n \circ s_n(x),x) \quad \forall \; x \neq \bar{x} \non
    \end{align}
    and for $x = \bar{x}$,
    \begin{align}
      \ut_n(g_n \circ \bar{s}_n(\bar{x}),\bar{x}) &=          \ut_n(g_n \circ s_n(\wi{x}),\bar{x})  \non \\
                                                  &= \ut_n(\wi{x},\bar{x}) \label{eq:u-xhat-xhat} \\
                                                  &\geq \ut_n(\bar{x},\bar{x}) =       \ut_n(g_n \circ s_n(\bar{x}),\bar{x}). \non
    \end{align}
    Here \eqref{eq:u-xhat-xhat} follows since $g_n \circ s_n(\wi{x}) = \wi{x}$, which in turn holds since $ \wi{x}\in \Dscr(g_n,s_n) $.

    Now, for all $x \in \Dscr(g_n,s_n)\backslash \{\bar{x}\}$,  $g_n \circ \bar{s}_n(x) = g_n \circ s_n(x) =x$ and hence $x$ lies in $ \Dscr(g_n,\bar{s}_n)$ and $\Dscr(g_n,s_n)$. However, when $x = \bar{x}$,  $g_n \circ \bar{s}_n(\bar{x}) = \wi{x} \neq \bar{x} = g_n \circ s_n(\bar{x})$. Thus, the sequence $\bar{x}$ lies in $\Dscr(g_n,s_n)$ but is not recovered by the pair $(g_n,\bar{s}_n)$ and hence does not lie in $\Dscr(g_n,\bar{s}_n)$. Thus,   $ |\Dscr(g_n,\bar{s}_n)| < |\Dscr(g_n,s_n)| $. 
    However, this is a contradiction since $s_n \in \pbest(g_n)$. Thus, for all $s_n \in  \pbest(g_n)$, the set  $\Dscr(g_n,s_n) $ is an independent set in $\Gs^n$.
  \end{proof}

\subsection{Proof of Theorem~\ref{thm:equil-dec-set-noiseless}}
\label{appen:thm3.2}
\begin{proof}
  From Lemma~\ref{lem:dec-set-ind-set},  $ \Dscr(g_n,s_n) $ is an independent set in $ \Gs^n $ for all $ g_n $ and for all  $ s_n \in \best(g_n) $. This implies that $ R(g_n) \leq \alpha(\Gs^n)^{1/n} $ for all strategies $ g_n $. We now show that for Stackelberg equilibrium strategies $ g_n^* $,  $ R(g_n^*) = \alpha(\Gs^n)^{1/n} $.
  
  Consider an independent set $I^n$ in $\Gs^n$ such that $ |I^n| = \alpha(\Gs^n) $ and define a strategy $g_n$ for the receiver  as
  \begin{align}
    g_n(x) = \left\{
    \begin{array}{c l}
      x  & \mbox{if} \; x \in I^n \\
      \Delta & \mbox{if} \; x \notin I^n
    \end{array}
\right.. \label{eq:defn-gn}
  \end{align}
  Since $\Delta$ is never preferred by the sender, we can assume without loss of generality, that for all $s_n \in \best(g_n)$  and for all $x$, $g_n \circ s_n(x) \in I^n$ and hence  $  \Dscr(g_n,s_n) \subseteq  I^n$. We will now show that in fact the two sets are equal. 

Consider an $x \in I^n$. For any $s_n \in \best(g_n)$, the utility of the sender is 
\begin{align}
  \ut_n(g_n \circ s_n(x),x) = \ut_n(x',x) \non
\end{align}
for some $x' \in I^n$. Since $I^n$  is an independent set in $\Gs^n$,  $\ut_n(x',x) < 0 $ for all $ x' \in I^n, x' \neq x $. Since $x \in I^n$ was arbitrary, 
\begin{align}
  \ut_n(g_n \circ s_n(x),x) \leq 0 \quad \forall \; x \in I^n, \non
\end{align}
with equality if and only if $ g_n \circ s_n(x) = x $. Clearly, the optimal choice of $s_n$ for the sender, is such that $ s_n(x) = x $ for all $ x \in I^n $. Specifically, all the strategies $ s_n \in \best(g_n) $ are such that $ s_n(x) = x $ for all $ x \in I^n $. Thus, for all $s_n \in \best(g_n)$,  $ \Dscr(g_n,s_n) = I^n $ and hence $ R(g_n) = \alpha(\Gs^n)^{1/n} $. It follows that for all Stackelberg equilibrium strategies $ g_n^* $ of the receiver,  $ R(g_n^*) =  \alpha(\Gs^n)^{1/n} $.
\end{proof}

\subsection{Proof of existence of information extraction capacity}
\label{appen:thm4.1}

Before proving the existence, we prove the following lemma.
\begin{lemma} \label{lem:alphaGmn-geq-GmGn}
Let $ m,n \in \Nbb $. Consider a sender with utility $ \ut $ and the corresponding sender graph $ \Gs^n $. Then,
	\begin{align}
	\alpha(\Gs^{m+n}) \geq \alpha(\Gs^m)\alpha(\Gs^n).\non
	\end{align}
      \end{lemma}
      \begin{proof}
	Consider an independent set $ I^m$ in $\Gs^m$  and an independent set $I^n$ in  $\Gs^n$.  The claim will follow by showing  that $I^m \times I^n$ is an independent set in $\Gs^{m+n}$.
	
	Consider sequences $x, y \in \Xscr^{m+n}$ such that $x = (w^m,w^n)$, $y = (v^m,v^n)$, where $w^m,v^m \in I^m$ and $w^n,v^n \in I^n$, with $w^k := (x_1,\hdots,x_k)$ and $v^k := (y_1,\hdots,y_k)$ for all $k$. Now,
	\begin{align}
          & \ut_{m+n}(y,x)  = \frac{m}{m+n} \ut_{m}(v^{m},w^{m}) + \frac{n}{m+n}\ut_{n}(v^{n},w^{n}). \non
	\end{align}
Since $I^m$ and $I^n$ are independent sets and $ x,y $ are distinct, $\ut_{m}(v^{m},w^{m}) \leq 0$ and $  \ut_{n}(v^{n},w^{n}) \leq 0$ with strict inequality in at least one the terms and hence,  $\ut_{m+n}(y,x) < 0$. This holds for all distinct sequences $x, y \in I^m \times I^n$ which shows that $I^m \times I^n$ is an independent set in $\Gs^{m+n}$. Thus,   $ \alpha(\Gs^{m+n}) \geq |I^m||I^n| $. Taking $ I^m $ and $ I^n $ to be the corresponding largest independent sets from $ \Gs^m $ and $ \Gs^n $ respectively, the claim follows.
\end{proof}

\begin{theorem}  \label{thm:exist-stra-cap}
	Consider a sender with utility $\ut$ and let $ \{\Gs^n\}_{ n\geq 1 } $ be the corresponding sequence of sender graphs. Then the limit in Definition~\ref{defn:inf-ext-cap} exists.
\end{theorem}
\begin{proof}
	From Lemma~\ref{lem:alphaGmn-geq-GmGn},  $\alpha(\Gs^{m+n}) \geq \alpha(\Gs^m)\alpha(\Gs^n)$,  for all $m, n \in \Nbb$. Define $\beta_n = \log(\alpha(\Gs^n))$ to get $\beta_{m+n} \geq \beta_{m} + \beta_m$. From Fekete's lemma \cite{schrijver2003combinatorial},  the limit of the sequence $\{\beta_n/n\}_{n \geq 1}$ exists and is equal to $ \sup_{n} \beta_n/n$. Using this and from the continuity and monotonicity of $\exp(.)$, 
	\begin{align}
	\lim_n  \;\exp \left(\frac{\beta_n}{n} \right)  = \sup_{n} \;\exp\left(\frac{\beta_n}{n}\right). \non
	\end{align}
	Substituting $\beta_n = \log(\alpha(\Gs^n))$, the claim follows.
      \end{proof}

      \subsection{Proof of Theorem~\ref{thm:generalization-shannon-cap}}
      \label{appen:thm4.2}

\begin{proof}
  Any pair of distinct symbols $ i ,j $ are adjacent in $ \Gs $ if and only if $ \ut(i,j) = 0 $. Further, for $ i \neq j $, $ \ut(i,j) = 0 $ if and only if $ \Ascr(i,j) = 1 $. Thus, distinct symbols $ i ,j $ are adjacent in $ \Gs $ if and only if $ i,j $ are adjacent in $ G $ and hence $ \Gs = G $. We now use this to show $ \Gs^n = G^{\boxtimes n } $.

  Consider two distinct sequences $ x,y \in \Xscr^n $.
  \begin{align}
    x \sim y \in \Gs^n &\Leftrightarrow \; \forall \;  k, \ut(x_k,y_k) = 0 \non \\
                       &\Leftrightarrow \; \forall \;  k, x_k = y_k \mbox{ or } x_k \sim y_k \in \Gs \non \\
                 &\Leftrightarrow \; x \sim y \in G^{\boxtimes n }. \non 
  \end{align}
  Thus, $ \Gs^n = G^{\boxtimes n } $ and hence $ \capacity = \Theta(G) $.
\end{proof}

\section{Lower Bounds}
\subsection{Proof of Theorem~\ref{thm:cap-lower-bound}}
\label{appen:thm5.2}

To prove Theorem~\ref{thm:cap-lower-bound}, we first define a lemma. For that consider the following definition
\begin{definition}\label{defn:Type-class-TKY}
Let $K \in \Nbb $ and $ \Yscr \subseteq \Xscr $. Define the set $ \TypeK $ as
\begin{align}
\TypeK = \curly{ x \in \Yscr^{K|\Yscr|} : P_x(i) = \frac{1}{|\Yscr|} \quad \forall \; i \in \Yscr}.  \non
\end{align}
\end{definition}
Thus, the set $ \TypeK $ is a set of all those sequences where every symbol from the set $ \Yscr $ occurs exactly $ K $ times. 


The following lemma gives a sufficient condition in terms of the optimization problem $ \Oscr(\ut) $ for the independence of the set $ \TypeK $.
\begin{lemma}\label{lem:opt-p-Tkw-ind}
Let $ n, K \in \Nbb $. Consider a sender with  utility $ \ut$ and let $ \Gs^n $ be the corresponding sender graph. Let $ \Yscr \subseteq \Xscr $ be a set feasible for the problem $ \Oscr(\ut) $. Then,  $\TypeK$ is an independent set in the graph $\Gs^{K|\Yscr|}$.
\end{lemma}
\begin{proof}
Fix a $ K \in \Nbb $. For distinct sequences $x,y \in \TypeK$,
  \begin{align}
    \ut_{K|\Yscr|}(y,x) &= \sum_{k \in [K|\Yscr|]} \frac{\ut(y_k,x_k)}{K|\Yscr|} = \sum_{i,j \in \Yscr} P_{x,y}(i,j)\ut(i,j), \non
  \end{align}
where the joint empirical distribution $ P_{x,y} $ satisfies   
\begin{align}
\sum_{i \in \Yscr}P_{x,y}(i,j) = \sum_{j \in \Yscr}P_{x,y}(i,j) = \frac{1}{|\Yscr|}, \non
\end{align}
for all $ i,j \in \Yscr $.

 Using the Birkhoff-von Neumann theorem \cite{schrijver2003combinatorial},  $ P_{x,y} $ can be given as a convex combination of the permutation matrices $ \Qscr = \{Q^{(0)},\hdots,Q^{(|\Qscr|)}\} $, \ie,  
  \begin{align}
    P_{x,y} = \frac{1}{|\Yscr|} \sum_{m} \alpha_m Q^{(m)}, \non
  \end{align}
  where $ \sum_{m}\alpha_m = 1$ and $ \alpha_m \geq 0 \;\forall \;m$. Since $ x $ and $ y $ are distinct, it follows that $ \alpha_0 < 1 $, where $ \alpha_0 $ corresponds to the coefficient of the identity matrix $ Q^{(0)} $. Thus, 
  \begin{align}
    \ut_{K|\Yscr|}(y,x)
                         &= \sum_{i,j \in \Yscr} P_{x,y}(i,j)\ut(i,j) \non \\
    &=  \frac{1}{|\Yscr|} \sum_{m \in [|\Qscr|]} \alpha_m\sum_{i,j \in \Yscr}  Q^{(m)}(i,j)\ut(i,j). \non
  \end{align}

  Since $ \Yscr$ is feasible for $ \Oscr(\ut) $, the term $ \sum_{i,j \in \Yscr}  Q^{(m)}(i,j)\ut(i,j) $ is negative for all $ m \in [|\Qscr|] $.  Thus, $ \ut_{K|\Yscr|}(y,x) < 0 $ for all distinct $ x, y \in \TypeK $ and hence, $ \TypeK $ is an independent set in the graph $ \Gs^{K|\Yscr|} $.
\end{proof}

We now prove Theorem~\ref{thm:cap-lower-bound}. \\
\begin{proof}[\textit{of Theorem~\ref{thm:cap-lower-bound}}]
  	Consider a set $ \Yscr^* \in \argmax_{\Yscr} \Oscr(\ut) $.  From Lemma~\ref{lem:opt-p-Tkw-ind}, it follows that $ T_{\Yscr^*}^K $ is an independent set in $ \Gs^{K|\Yscr^*|} $ for all $ K \in \Nbb $ and hence $ \alpha(\Gs^{K|\Yscr^*|}) \geq |T_{\Yscr^*}^K| $. Taking the limit and applying Stirling's approximation gives that  $\capacity \geq |\Yscr^*| =  \Gamma(\ut) $.
\end{proof}

\subsection{Proof of Proposition~\ref{prop:i1-i2-condi-for-Y}}
\label{appen:prop5.5}
For ease of explanation, we define the set of symbols satisfying \eqref{eq:neg-weight-chain} as a  negative-weight chain.
\begin{definition}[Negative-weight chain]\label{defn:neg-wt-cycle}

        A sequence of distinct symbols  $ i_0,i_1 ,\hdots, i_{K-1} \in \Yscr $ form a negative-weight chain if   
  \begin{align}
    \ut(i_1,i_0) + \ut(i_2,i_1) + \hdots + \ut(i_0,i_{K-1}) < 0. \non
  \end{align}
\end{definition}
Notice that the chain in the above definition is a closed-chain since the first and last symbols are the same. For brevity, we describe a closed-chain as simply a chain.

      \begin{proof}
  Suppose a set $ \Yscr \subseteq \Xscr $ is such that any sequence of distinct symbols $ i_0,i_1 ,\hdots, i_{K-1} \in \Yscr $ form a negative-weight chain. We show that $ \Yscr $ is feasible for $ \Oscr(\ut) $. 

  Consider a permutation matrix $ Q \in \Qscr \setminus \{\mathbf{I}\} $. 
  For this permutation matrix $ Q $, define a permutation $ \pi : \Yscr \rarr \Yscr $ as $ i = \pi(j) $  if and only if $ Q(i,j) = 1 $.
  
  Now, a permutation can be decomposed into disjoint cycles. Thus, a permutation $ \pi : \Yscr \rarr \Yscr $ can be represented as 
\begin{align}
  (i_0^1  \hdots i_{K_1-1}^1) \hdots (i_0^M  \hdots i_{K_M-1}^M), (v_1)  \hdots (v_R), \label{eq:decom-in-disjoint-cycles}
\end{align}
where $ i_k^j, v_k \in \Yscr $ are all distinct symbols and $ M ,K_m, R \in \Nbb, K_m \geq 2 \;\forall \;m $, are such that $ \sum_{m \in [M]}K_m + R = |\Yscr| $.  The above arrangement represents that for all $ m \in [M] $, the symbols $ \{i_0^m,\hdots,i_{K_m-1}^m \} $ form cycles where $ \pi(i_k^m) = i_{(k+1) \mbox{mod }K_m}^m$. Further,  $ \pi(v_k) = v_k $ for all $ k \in [R]$. 


Using this  representation in  \eqref{eq:decom-in-disjoint-cycles}, we write
\begin{align}
  & \sum_{i,j \in \Yscr}  Q(i,j)\ut(i,j) =  \sum_{j \in \Yscr}  \ut(\pi(j),j) \non \\
  &= \ut(i_1^1,i_0^1) + \ut(i_2^1,i_1^1) + \hdots + \ut(i_0^1,i_{K_1-1}^1) + \hdots \non \\
  & + \ut(i_1^M,i_0^M) + \ut(i_2^M,i_1^M) + \hdots + \ut(i_0^M,i_{K_M-1}^M) \non \\
  &< 0. \non
\end{align}
The last inequality follows since any sequence of distinct symbols from the set $ \Yscr $ form a negative-weight chain. The above relation is true for all permutation matrices $ Q \in \Qscr \setminus \{\mathbf{I}\}$.  Hence, the set $ \Yscr $ is feasible for $ \Oscr(\ut) $.

Suppose the set $ \Yscr $ is feasible for $ \Oscr(\ut) $. Take distinct symbols $ i_0,i_1 ,\hdots, i_{K-1} \in \Yscr $  for some $ K \in [|\Yscr|], \; K \geq 2 $ and choose a permutation $ \pi' : \Yscr \rarr \Yscr $ such that  $ i_0,i_1 ,\hdots, i_{K-1} $ form a cycle as $ \pi'(i_{k}) = i_{(k+1) \mbox{mod} K}$. The rest of the symbols are mapped to themselves. Let $ Q^{\pi'} $ be the corresponding permutation matrix. Observe that $ Q^{\pi'} $ is not an identity matrix. Then, it follows from feasibility of $ \Yscr $ that
\begin{align}
  &\sum_{i,j \in \Yscr }Q^{\pi'}(i,j)\ut(i,j) \non \\
  &=   \ut(i_1,i_0) + \ut(i_2,i_1) + \hdots + \ut(i_0,i_{K-1}) < 0, \non
\end{align}
and hence the symbols form a negative-weight chain. Since the symbols were chosen arbitrarily, this holds for all sequences of distinct symbols from the set $ \Yscr $. This completes the proof.
\end{proof}

\subsection{Proof of Theorem~\ref{thm:suff-cond-for-pos-LP}}
\label{appen:thm7.2}

Recall the definition of a positive-edges cycle in Definition~\ref{defn:pos-wt-cycle}. The following lemma demonstrates that it is necessary that a set $ \Yscr $ does not contain a positive-edges cycle for it to be feasible for $ \Oscr(\ut)$. 
\begin{lemma} \label{lem:direc-cycle-LP-neg}
	Consider a sender with utility $ \ut $ and let $ \Gs $ be the corresponding sender graph. Consider a $ \Yscr $ which contains a positive-edges cycle. Then,  $ \Yscr $ is  not feasible for $ \Oscr(\ut) $.
      \end{lemma}
      \begin{proof}
  Since there  exists a positive-edges cycle in $ \Yscr $, there is a sequence of distinct symbols  $ i_0,i_1 ,\hdots, i_{K-1} \in \Yscr $ such that   
  \begin{align}
    \ut(i_l,i_m) \geq 0 \quad \forall \; \; l = (m+1) \mbox{mod}\; K. \non 
  \end{align}
  Clearly, this means that the set $ \Yscr $ does not satisfy the condition given in Proposition~\ref{prop:i1-i2-condi-for-Y}.
\end{proof}

\begin{proof}[\textit{of Theorem~\ref{thm:suff-cond-for-pos-LP}}]
  Let $ i_0,\hdots,i_{K-1} $ be distinct symbols from $ \Yscr $.   Since there is no positive-edges cycle in $ \Yscr $, there is at least one pair of $ (i_k,i_j) $, $ k > j $  such that  $ \ut(i_k,i_j) < 0 $. Then, 
  \begin{align}
    &\ut(i_1,i_0) + \ut(i_2,i_1) + \hdots + \ut(i_0,i_{K-1}) \non \\
    &\leq \min_{i,j \in \Yscr: \ut(i,j) < 0} |\ut(i,j)| \non \\
    &+ (|\Yscr|-1)\max_{i,j \in \Yscr: \ut(i,j) \geq 0} \ut(i,j) \non \\
    &< 0. \non 
  \end{align}
  Since $ i_0,\hdots,i_{K-1} $ were arbitrary, using Proposition~\ref{prop:i1-i2-condi-for-Y}, it follows that $ \Yscr $ is feasible for $ \Oscr(\ut) $.
\end{proof}


\subsection{Proof of Corollary~\ref{coro:Y-is-ind-set}}
\label{appen:coro5.7}

    \begin{proof}
  If $ \Yscr $ is feasible for $ \Oscr(\ut) $, then
 \begin{align}
    \ut(i_2,i_1) + \ut(i_1,i_2)  < 0 \quad \forall \;  i_1, i_2 \in \Yscr, i_1 \neq i_2. \non
  \end{align}
This implies that for all distinct symbols $ i_1,i_2 \in \Yscr $, we have $ \utsym(i_1,i_2) = \utsym(i_2,i_1) < 0 $.  Thus, $ \Yscr $ is an independent set in $ \Gsym $. The upper bound on $ \Gamma(\ut) $ is now obvious.
\end{proof}

\subsection{Proof of Theorem~\ref{thm:limit_of_heirarchy}}
\label{appen:thm5.12}

First, consider a multiplicative property of $ \Gamma(\ut_n) $ analogous to Lemma~\ref{lem:alphaGmn-geq-GmGn}. 
   \begin{lemma} \label{lem:gamma-m+n-gamma-m-times-n}
Let $ m,n \in \Nbb $. For any utility $ \ut $,
     \begin{align}
       \Gamma(\ut_{m+n}) \geq \Gamma(\ut_{m})\Gamma(\ut_{n}). \non
     \end{align}
   \end{lemma}
 \begin{proof}
 Consider two sets $ \Yscr_m, \Yscr_n $ that are optimal for $ \Oscr(\ut_m) $ and $ \Oscr(\ut_n) $ respectively.
     Define $ \Yscr^{m+n} = \Yscr_m \times \Yscr_n $. For any $ y \in \Yscr^{m+n} $, we denote $ y = (\ybar, \ytilde) $, $ \ybar \in \Yscr_m, \ytilde \in \Yscr_n$. Moreover, notice that for any $ y,z \in \Yscr^{m+n} $
\[ \ut_{m+n}(y,z) = \frac{m}{m+n}\ut_m(\ybar,\zbar) +  \frac{n}{m+n}\ut_n(\ytilde,\ztilde).\]
      Now, consider distinct sequences  $ y^1, y^2, \hdots, y^L  \in \Yscr^{m+n} $. Using the above, we write
          \begin{align}
            &   \ut_{m+n}(y^2,y^1) +        \ut_{m+n}(y^3,y^2) + \hdots +        \ut_{m+n}(y^1,y^L) \non \\
            & = \frac{m}{m+n}  \Big(       \ut_m(\ybar^2,\ybar^1) +        \ut_m(\ybar^3,\ybar^2) + \hdots +        \ut_m(\ybar^1,\ybar^L) \Big) \non \\
            &+ \frac{n}{m+n} \Big(       \ut_n(\ytilde^2,\ytilde^1) +        \ut_n(\ytilde^3,\ytilde^2) + \hdots +        \ut_n(\ytilde^1,\ytilde^L) \Big). \non
          \end{align}
          Suppose there exists some  $ l_1 < l_2 $ such that  $ \ybar^{l_1} = \ybar^{l_2} $. Suppose all the sequences $ \ybar^{l_1},\ybar^{l_1 + 1},\hdots,\ybar^{l_2-1} $ between $ l_1$  and $ l_2 $ are distinct. Using the analogous result of Proposition~\ref{prop:i1-i2-condi-for-Y} for $n$-length sequences, it can be observed that $ \ybar^{l_1},\ybar^{l_1 + 1},\hdots,\ybar^{l_2-1} $ form a negative-weight chain. Thus, removing the edges corresponding to this chain we get
          \begin{align}
            &             \ut_m(\ybar^2,\ybar^1) +  \ut_m(\ybar^3,\ybar^2) + \hdots +        \ut_m(\ybar^1,\ybar^L) \non \\
            &\leq   \sum_{l \in [L] \setminus \{l_1,\hdots,l_2-1\}}      \ut_m(\ybar^{ (l+1) \mbox{mod } L},\ybar^l). \non 
          \end{align}
          We remove all such negative-weight chains from $ \{\ybar^l \}_{l \in  [L]}$ and $ \{\ytilde^l \}_{l \in [L]} $ to get a pair of sets of distinct sequences $  \{\ybar^l \}_{l \in  [\Lbar]} $  and $ \{\ytilde^l \}_{l \in  [\Ltilde]} $ where $ \Lbar, \Ltilde \subseteq [L] $. Since $ y^1, \hdots, y^L $ are all distinct sequences, either $ \Lbar \neq \emptyset $ or $ \Ltilde \neq \emptyset $. Moreover, since $ \{\ybar^l \}_{l \in  [\Lbar]} \subseteq \Yscr_m $  and $ \{\ytilde^l \}_{l \in  [\Ltilde]}  \subseteq \Yscr_n $, the sequences form negative-weight chains and hence we get 
                    \begin{align}
            &   \ut_{m+n}(y^2,y^1) +        \ut_{m+n}(y^3,y^2) + \hdots +        \ut_{m+n}(y^1,y^L) \non \\
                      &\leq  \frac{m}{m+n}   \sum_{l \in \Lbar}      \ut_m(\ybar^{ (l+1) \mbox{mod } L},\ybar^l) \non  \\
        &             + \frac{n}{m+n}  \sum_{l \in \Ltilde}      \ut_n(\ytilde^{ (l+1) \mbox{mod } L},\ytilde^l) < 0. \non 
          \end{align}
This is true for all sets of distinct sequences from $\Yscr^{m+n}$.   Thus, $ \Yscr^{m+n} $ is feasible for $ \Oscr(\ut_{m+n}) $ and hence 
          \begin{align}
\Gamma(\ut_{m+n}) \geq |\Yscr^{m+n}|  = \Gamma(\ut_m)\Gamma(\ut_n). \non
          \end{align}

        \end{proof}


\begin{proof}[\textit{of Theorem~\ref{thm:limit_of_heirarchy}}]
  From Corollary~\ref{corr:ind-set-feas-for-O}, it follows that all independent sets of $ \Gs^n $ are feasible for $ \Oscr(\ut_n) $. Thus, $ \Gamma(\ut_n) \geq \alpha(\Gs^n) $ and hence $ \lim_{n \rarr \infty} \Gamma(\ut_n)^{1/n}  \geq \capacity $.
  
   Now we prove that  $  \lim_{n \rarr \infty}\Gamma(\ut_n)^{1/n} \leq \capacity  $.    Observe that Lemma~\ref{lem:opt-p-Tkw-ind} does not impose any specific structure for the set $ \Xscr $. In particular, in place of the  set $ \Xscr $ we can consider its $n$-fold Cartesian product $ \Xscr^n $ and the results will hold.

Let $ K \in \Nbb. $  Consider a set $ \Yscr_n $ that  maximizes $ \Oscr(\ut_n) $.  Following Definition~\ref{defn:Type-class-TKY}, we define a set $ T_{\Yscr_n}^{K} $ which consists of sequences constructed by concatenating sequences from $ \Yscr_n $, each appearing exactly $ K $ times. This construction gives sequences of length  $ nK|\Yscr_n| $. Analogous to Lemma~\ref{lem:opt-p-Tkw-ind} it follows that $ T_{\Yscr_n}^{K} $  is an independent set in $ \Gs^{nK|\Yscr_n|} $ for all $ K \in \Nbb $. Consequently, the capacity is bounded as 
  \begin{align}
    \capacity &= \lim_{K \rarr \infty} \alpha(\Gs^{nK|\Yscr_n|})^{\frac{1}{nK|\Yscr_n|}} \non \\
    &\geq   \lim_{K \rarr \infty}  \Big(|T_{\Yscr_n}^{K}|^{\frac{1}{ K|\Yscr_n|}}\Big)^{1/n}. \non 
  \end{align}
We use  Stirling's approximation to determine $ \lim_{K \rarr \infty}|T_{\Yscr_n}^{K}|^{1/ K|\Yscr_n|} $. Notice that $ |T_{\Yscr_n}^{K}| $ is now given as $  |T_{\Yscr_n}^{K}| = (K |\Yscr_n|)!/(K!)^{|\Yscr_n|}$.  Thus, 
  \begin{align}
\capacity  \geq   \lim_{K \rarr \infty}  \Big(|T_{\Yscr_n}^{K}|^{\frac{1}{ K|\Yscr_n|}}\Big)^{1/n} = |\Yscr_n|^{1/n} = \Gamma(\ut_n)^{1/n}. \non 
  \end{align}
  Taking the limit as $ n \rarr \infty $, the claim follows.
             \end{proof}

  \subsection{Computational approach to check for the feasibility of a set $ \Yscr$}
\label{appen:comp_approach}
  We now discuss a computational approach that can be used to check for the feasibility of a set $ \Yscr $ for the problem $ \Oscr(\ut)$.
  Given a set of vertices $ \Yscr $, its feasibility for $ \Oscr(\ut) $ can be checked using the Bellman-Ford algorithm in the following way. Define a weighted directed graph $ \bar{\Gs} $ with vertices $ \Xscr $ induced by $ \ut $ where every vertex is connected to every other vertex and the edge from a vertex $ i $ to $ j $ has weight $ \ut(j,i) $. Define a utility $ \ut' = - \ut $ and let the directed graph induced by $ \ut' $ be denoted as $ \Gs' $. Suppose we have a set $ \Yscr $ that is feasible for the problem $ \Oscr(\ut) $. Then, from the characterization of feasible region of $ \Oscr(\ut) $, it implies that there is no zero or positive-weight directed cycle in the subgraph of $ \bar{\Gs} $ induced by  $ \Yscr $. Equivalently, there is no negative-weight or zero-weight directed cycle in the subgraph of $ \Gs' $ induced by $ \Yscr $. We use this observation as follows.

  The Bellman-Ford algorithm (\cite{cormen2009introduction}, Ch. 26), determines the shortest paths to all vertices in a directed graph from a given source vertex. It is also known that the algorithm detects whether the graph has a negative-weight directed cycle in the graph, in which case there may not exist a shortest path between two vertices. Suppose there is no zero-weight directed cycle in the graph $ \Gs' $. Then, this algorithm can be used to check for the feasibility of a given $ \Yscr $ in the following manner. 
  \begin{itemize}
  \item Given a set $ \Yscr $ consider the subgraph of $ \Gs' $ induced by the vertices $ \Yscr $
  \item Apply the Bellman-Ford algorithm on this subgraph 
  \item If the algorithm detects a negative-weight directed cycle in the subgraph, then the set $ \Yscr $ is not feasible for $ \Oscr(\ut) $. Otherwise, $ \Yscr $ is feasible for $ \Oscr(\ut) $.

  \end{itemize}

  In the worst case, we need to check feasibility of all subsets of $ \Xscr $ in order to determine the optimal $ \Yscr $. However, the total number of subsets of  $ \Xscr $ may be very large which may make this procedure impractical.  
  Nevertheless, the above method is still independent of the blocklength $ n $ and is thus an effective tool to derive a lower bound for the capacity.

\section{Upper Bounds}   

   \subsection{Proof of Lemma~\ref{lem:symm-util-upper-bnd}}
\label{appen:lem6.1}
    \begin{proof}
Let $ n \in \Nbb $.     Consider distinct sequences $ x,y \in  \Xscr^n $ which are adjacent in $ (\Gsym)^n $. Thus,  $ \utsym_n(y,x) = \utsym_n(x,y) =  (\ut_n(y,x) + \ut_n(x,y))/2 \geq  0 $, which gives that either $ \ut_n(y,x) \geq 0 $  or $ \ut_n(x,y) \geq 0 $. In either case, $ x,y $ are  adjacent in $ \Gs^n $ and hence $ (\Gsym)^n$ is a subgraph of $\Gs^n $. This gives $ \alpha(\Gs^n) \leq  \alpha((\Gsym)^n) $. The claim follows by taking the limit of the $n$-th root.
\end{proof}

\subsection{Proof of Theorem~\ref{thm:upper-bnd-Shannon-cap-symm}}
\label{appen:thm6.2}
\begin{proof}
 Let $ n \in \Nbb $. We first prove that if $ \ut $ is  symmetric then $ \Xi(\ut) \leq \Theta(\ut) $. Consider the graph $ \Gs^{\boxtimes n} $ derived by taking the $n$-fold strong product of $ \Gs $. Consider a distinct pair of sequences $ x,y \in \Xscr^n $ that are adjacent in $ \Gs^{\boxtimes n} $. Then, for all $ k \in [n] $, either $x_k = y_k $ or $ x_k, y_k$ are  adjacent in $ \Gs $. Thus, 
\begin{align}
\ut_n(y,x) = \frac{1}{n}\sum_{y_k \neq x_k} \ut(y_k,x_k) \geq 0. \non
\end{align}
Hence, $ x, y $  are adjacent in $ \Gs^n $ as well, which gives $ \Gs^{\boxtimes n}$ is a subgraph of $\Gs^n $. Thus, $ \alpha(\Gs^n)\leq \alpha(\Gs^{\boxtimes n})$ and taking the limit of the $n$-th root and using  Lemma~\ref{lem:symm-util-upper-bnd}, the claim follows.
\end{proof}

\subsection{Proof of Theorem~\ref{thm:posi-diff-greater}}
\label{appen:thm6.4}
\begin{proof}
Let $ n \in \Nbb $. We prove this by showing that if distinct sequences $ x,y \in \Xscr^n $ are not adjacent in $ \Gs^n $, then $ x, y $ are not adjacent in $ \Gs^{\boxtimes n} $.

    Let $ x, y $ be not adjacent in $ \Gs^n $. Then, $ \ut_n(y,x) < 0 $ and $ \ut_n(x,y) < 0 $. Since, $b \leq a$, it follows that
    \begin{align}
      \Big|\{k : \ut(y_k,x_k) = -b\}\Big| &> \frac{n}{2}, \non \\
      \;\; \Big|\{k : \ut(x_k,y_k) = -b\}\Big| &> \frac{n}{2}. \non
    \end{align}
 This implies that there is some $k$ such that $\ut(y_k,x_k) = \ut(x_k,y_k) = -b$ and hence  $x_k, y_k $  are not adjacent in $ \Gs $. It follows that $ x, y  $ are not adjacent in $ \Gs^{\boxtimes n} $ as well.  This gives that $  \alpha(\Gs^n) \leq \alpha(\Gs^{\boxtimes n}) $. 
 The claim follows by taking the limit.
  \end{proof}

\subsection{Proof of Theorem~\ref{thm:exact-charac-of-capa}}
\label{appen:thm6.5}
   
\begin{proof}
  \begin{enumerate}

    \item Follows from Theorem~\ref{thm:upper-bnd-Shannon-cap-symm} and the fact that for a perfect graph, $ \Theta(\Gs) = \alpha(\Gs) $ \cite{lovasz1979shannon}. 

      \item  Follows from Theorem~\ref{thm:posi-diff-greater} and the fact that for a perfect graph, $ \Theta(\Gs) = \alpha(\Gs) $.
  \end{enumerate}
\end{proof}

\section{Noisy Channel Results}
\label{appen:noisy}

Before we prove the results, we state a few definitions. For any strategy $g_n$ of the receiver, consider $ \pbest(g_n) $ defined as 
\begin{align}
   \pbest(g_n) = \argmin_{s_n \in \best(g_n)} |\Dscr(g_n,s_n)|.  \label{eq:worst-best-resp}
\end{align}
Thus, $ \pbest(g_n) $ is the set of best responses which give the least objective value to the receiver. Let $$\Zscr(y) = \supp(P_{Z|Y}(\cdot | y)),$$ where $y$ is an input to the channel. Note that since the output space of the channel is $\Xscr^n$, $\Zscr(y) \subseteq \Xscr^n $ for all $ y $.
\subsection{Proof of Theorem~\ref{thm:stac-eq-recov-ind-no}}
\label{appen:thm8.1}
\begin{proof}
Consider a Stackelberg equilibrium strategy $ g_n^* $ of the receiver.  It is known from Lemma~\ref{lem:dec-set-ind-set}, that the set of recovered sequences from the sender can be at most $ \alpha(\Gs^n) $ and hence \newline $ \min_{s_n \in \best(g_n^*)}|\Dscr(g_n^*,s_n)| \leq \alpha(\Gs^n) $. Further, at most $ \alpha(\Gc^n) $ sequences can be transmitted with zero error through the channel and hence for all $ s_n $, $ |\Dscr(g_n^*,s_n)| \leq \alpha(\Gc^n) $. Together, it follows that
    \begin{align}
  \min_{s_n \in \best(g_n^*)}|\Dscr(g_n^*,s_n)| \leq \min\{ \alpha(\Gs^n),\alpha(\Gc^n)\}. \non  
    \end{align}
    We now show that equality holds in the above relation.

 Let $ d = \min\{ \alpha(\Gs^n),\alpha(\Gc^n)\} $. Clearly, there exists an independent set $ I_{\mathsf{s}}^n $ in $ \Gs^n $ such that $ |I_{\mathsf{s}}^n| = d $. Similarly, there exists an independent set $ I_{\mathsf{c}}^n $ in $ \Gc^n $ such that $ | I_{\mathsf{c}}^n| = d$. 
  Let the sequences in the sets $I_{\mathsf{s}}^n$ and $I_{\mathsf{c}}^n$ be denoted as  $x^i$ and $y^i$ respectively, with  $i \in [d]$ and $x^i, y^i \in \Xscr^n$. With this convention, define the strategy $g_n$ as
  \begin{align}
g_n(z)  = \left\{
    \begin{array}{c l}
      x^i & \mbox{if} \; z \in \Zscr(y^i)  \\
      \Delta & \mbox{if} \; z \notin \bigcup_{i=1}^{d} \Zscr(y^i) 
    \end{array}
\right.. \label{eq:str-of-rec}
  \end{align}
  We show that the strategy $ g_n $ of the receiver ensures that all strategies $s_n$ in the set of best responses $ \best(g_n) $ are such that $\Zscr(s_n(x^i)) = \Zscr(y^i)$ for all $i \in [d]$. Fix an index $ i $, let $ s_n $ be any strategy for the sender and let     $s_n(x^i) = y^* \in \Xscr^n$. Notice that if $ \Zscr(y^*) \nsubseteq \bigcup_{j=1}^{d} \Zscr(y^j)$, then $g_n(z) = \Delta$ for some   $z \in \Zscr(y^*)$, $ z \notin \bigcup_{j=1}^{d} \Zscr(y^j)$. This gives that $\expec{\ut_n(g_n(Z),x^i)} = -\infty $.  Thus, $ y^* $ is such that $ \Zscr(y^*) \subseteq \bigcup_{j=1}^{d} \Zscr(y^j) $. 

  Writing $\Zscr(y^*) = \bigcup_{j=1}^{d}\Zscr(y^*) \cap \Zscr(y^j)$, it follows that
    \begin{align}      
      &    \expec{\ut_n(g_n(Z),x^i)}  \non \\
                &=\sum_{z \in \;\bigcup_{j=1}^{d}\Zscr(y^*) \cap \Zscr(y^j)} P_{Z|Y}(z|y^*)\ut_n(g_n(z),x^i) \non \\
          &= \sum_{j=1}^{d}\sum_{z \in \Zscr(y^*) \cap \Zscr(y^j)} P_{Z|Y}(z|y^*)\ut_n(x^j,x^i). \non
    \end{align}
The last equation follows from the definition of $g_n$. Since $I_{\mathsf{s}}^n$ is an independent set in $\Gs^n$,  $\ut_n(x^j,x^i) < \ut_n(x^i,x^i) = 0 $ for all $j \neq i$ and hence
    \begin{align}
&      \sum_{j=1}^{d}\sum_{z \in \Zscr(y^*) \cap \Zscr(y^j)} P_{Z|Y}(z|y^*)\ut_n(x^j,x^i) \leq 0 \non
    \end{align}
    with equality if and only if $ \Zscr(y^*) = \Zscr(y^i) $. Thus, the strategy $ y^* $ is such that  $ \Zscr(y^*) = \Zscr(y^i) $. Clearly, this holds for all $ s_n \in \best(g_n) $ and for all  $ i \in [d] $.

It easy to see that the utility of the sender and the receiver do not depend on the exact choice of $ y^* $ so long as $ \Zscr(y^*)=\Zscr(y^i) $. Hence, without loss of generality, we consider $ s_n $ to be such that $s_n(x^i) = y^i$ for all $ i \in [d] $. 
    Thus, when the sequence $x^i \in I_{\mathsf{s}}^n$ is observed by the sender, it encodes it as $y^i$. The channel generates an output $z$, which belongs to the support set $\Zscr(y^i)$. The receiver maps all such $z$ to $x^i$ thereby ensuring $\Pbb(\wi{X} = x^i | X = x^i) = 1$ for all $i \in [d]$. Thus, for all $x^i \in I_{\mathsf{s}}^n$ and  $s_n \in \best(g_n)$,  $x^i \in \Dscr(g_n,s_n)$. Hence $\Dscr(g_n,s_n) \supseteq I_{\mathsf{s}}^n$ for all $s_n \in \best(g_n)$. Since $|I_{\mathsf{s}}^n| = d$,  $|\Dscr(g_n,s_n)| \geq d$. Using $|\Dscr(g_n,s_n)| \leq d = \min\{\alpha(\Gs^n),\alpha(\Gc^n)\}$, it follows that for all Stackelberg equilibrium strategies $g_n^*$, $ \min_{s_n \in \best(g_n^*)}|\Dscr(g_n^*,s_n)| = \min \{\alpha(\Gs^n),\alpha(\Gc^n)\}$.
  \end{proof}
Thus, from the strategy defined in \eqref{eq:str-of-rec}, it can be observed that the receiver decodes meaningfully only for a subset of sequences. In particular, it chooses $d$ number of inputs, $y^i $ with $ i \; \in [d]$, which can be distinguished from each other and maps the respective support sets, $\Zscr(y^i)$, to $d$ distinct sequences $x^i$. For the rest of the outputs from the channel, the receiver declares an error $\Delta$. In response, the optimal strategy for the sender is (without loss of generality) such that it maps the sequences $x^i$ to the inputs $y^i$. Also, for all strategies $ s_n $ and for all $x \in \Xscr^n \setminus \{x^i\}_{ i \in [d] }$, $s_n(x) = y^*$, where $y^*$ is such that $\Zscr(y^*) \subseteq \bigcup_{j=1}^{d} \Zscr(y^j)$. This is because for any other input $y' \neq y^*$, if $\Zscr(y') \nsubseteq \bigcup_{j=1}^{d} \Zscr(y^j)$, then the channel output $z$ may lie outside $\bigcup_{j=1}^{d} \Zscr(y^j)$ for which $g_n(z) = \Delta$ and the corresponding utility is $-\infty$.

  \subsection{Proof of Theorem~\ref{thm:limit-D-min-of-capa}}
  \label{appen:thm8.2}
  \begin{proof}
  Let $\{g_n^*\}_{ n \geq 1}$ be a sequence of Stackelberg equilibrium strategies for the receiver. From Theorem~\ref{thm:stac-eq-recov-ind-no},
  \begin{align}
\min_{s_n \in \best(g_n^*)}|\Dscr(g_n^*,s_n)| = \min \{\alpha(\Gs^n),\alpha(\Gc^n)\}. \non 
  \end{align}
It can be shown that, for a discrete memoryless channel given by \eqref{eq:dmc-defn}, the graph $\Gc^n$ is same as the graph constructed by taking $n$-fold strong product of $\Gc$  \cite{shannon1956zero}, \ie, $     \Gc^n = \Gc^{\boxtimes n} $.
  Thus, $ \min_{s_n \in \best(g_n^*)}|\Dscr(g_n^*,s_n)| = \min \{\alpha(\Gs^n),\alpha(\Gc^{\boxtimes n})\}$  and hence
  \begin{align}
    R(g_n^*) &=   \min_{s_n \in \best(g_n^*)} |\Dscr(g_n^*,s_n)|^{1/n} \non \\
    &= \min\curly{ \alpha(\Gs^n)^{1/n},\alpha(\Gc^{\boxtimes n})^{1/n}}. \non
  \end{align}
The claim follows after taking the limit.
\end{proof}

\bibliography{ref.bib}

\end{document}